\documentclass[11pt]{article}
\usepackage{amsmath,amsthm,amssymb,bm,color,xspace,booktabs}

\usepackage[margin=1in]{geometry}
\usepackage{xcolor}
\usepackage[colorlinks=true]{hyperref}
\usepackage{graphicx}
\usepackage{comment}
\usepackage{algorithm}
\usepackage[noend]{algpseudocode}
\usepackage{algorithmicx}
\usepackage{slashbox}
\usepackage{url}

\algrenewcommand\textproc{\textsf}
\newtheorem{theorem}{Theorem}[section]
\newtheorem{lemma}[theorem]{Lemma}

\newtheorem{proposition}[theorem]{Proposition}
\newtheorem{claim}{Claim}

\theoremstyle{plain}
\newtheorem{definition}[theorem]{Definition}

 \usepackage{multirow}

\newcommand{\opt}{\mathrm{OPT}}
\newcommand{\bbN}{\mathbb{N}}
\newcommand{\bbR}{\mathbb{R}}
\newcommand{\bbZ}{\mathbb{Z}}
\newcommand{\optD}{\overline{\mathrm{OPT}}}
\newcommand{\set}[1]{\{#1\}}

\newcommand{\I}{\mathcal{I}}

\newcommand{\Dr}{\Delta_{\mathrm{r}}}
\newcommand{\Db}{\Delta_{\mathrm{b}}}

\usepackage[T1,T5]{fontenc}

\DeclareMathOperator*{\argmax}{arg\,max}

\newcommand\bsbox[2]{%
  \raisebox{-1pt}{%
  \stackinset{c}{}{c}{}{\rotatebox{45}{\rule[-.5em]{.2pt}{3em}}}{%
  \def\stacktype{L}%
  \setstackgap{L}{.5\baselineskip}%
  \kern-1pt\makebox[\widthof{\refA}][r]{#1}\stackon{}{\smash{\makebox[\widthof{\refB}]{#2}}}\kern1pt%
  }%
  \kern-5.5pt}%
}

\allowdisplaybreaks

\title{Approximability of Monotone Submodular Function Maximization under Cardinality and Matroid Constraints in the Streaming Model}

\author{
  Chien-Chung Huang  \\ CNRS, DI ENS, PSL \\  \texttt{villars@gmail.com}
  \and
  Naonori Kakimura\thanks{Supported by JST ERATO Grant Number JPMJER1201, Japan, and by JSPS KAKENHI Grant Number JP17K00028.}\\ Keio University\\  \texttt{kakimura@math.keio.ac.jp}
     \and
  Simon Mauras \\     Universit\'{e} de Paris, IRIF, CNRS \\  \texttt{simon.mauras@irif.fr}
    \and
  Yuichi Yoshida \\
  National Institute of Informatics\\
  \texttt{yyoshida@nii.ac.jp}
}

\begin{document}
\maketitle
\begin{abstract}
  Maximizing a monotone submodular function under various constraints is a classical and intensively studied problem.
  However, in the single-pass streaming model, where the elements arrive one by one and an algorithm can store only a small fraction of input elements,
  there is much gap in our knowledge, even though several approximation algorithms have been proposed in the literature.

  In this work, we present the first lower bound on the approximation ratios for cardinality and matroid constraints that beat $1-\frac{1}{e}$ in the single-pass streaming model.
  Let $n$ be the number of elements in the stream.
  Then, we prove that any (randomized) streaming algorithm for a cardinality constraint with approximation ratio $\frac{2}{2+\sqrt{2}}+\varepsilon$ requires $\Omega\left(\frac{n}{K^2}\right)$ space for any $\varepsilon>0$, where $K$ is the size limit of the output set.
  We also prove that any (randomized) streaming algorithm for a (partition) matroid constraint with approximation ratio $\frac{K}{2K-1}+\varepsilon$ requires $\Omega\left(\frac{n}{K}\right)$ space for any $\varepsilon>0$, where $K$ is the rank of the given matroid.

  In addition, we give streaming algorithms when we only have a weak oracle with which we can only evaluate function values on feasible sets.
  Specifically, we show weak-oracle streaming algorithms for cardinality and matroid constraints with approximation ratios $\frac{K}{2K-1}$ and $\frac{1}{2}$, respectively, whose space complexity is exponential in $K$ but is independent of $n$.
  The former one exactly matches the known inapproximability result for a cardinality constraint in the weak oracle model.
  The latter one almost matches our lower bound of $\frac{K}{2K-1}$ for a matroid constraint, which almost settles the approximation ratio for a matroid constraint that can be obtained by a streaming algorithm whose space complexity is independent of $n$.
\end{abstract}

\thispagestyle{empty}
\setcounter{page}{0}
\newpage


\section{Introduction}

A set function $f\colon2^E \rightarrow \mathbb{R}$ on a ground set $E$ is \emph{submodular} if it satisfies the \emph{diminishing marginal return property}, i.e., for any subsets $S \subseteq T \subsetneq E$ and $e\in E \setminus T$,
\[
  f(S \cup \set{e}) - f(S)\geq f(T \cup \set{e}) - f(T).
\]
A function is \emph{monotone} if $f(S)\leq f(T)$ for any $S\subseteq T \subseteq E$.
Submodular functions play a fundamental role in combinatorial optimization, as they capture rank functions of matroids, edge cuts of graphs, and set coverage, just to name a few examples.

In addition to their theoretical interests, submodular functions have attracted much attention from the machine learning community because they can model various practical problems such as online advertising~\cite{Alon:2012em,Kempe:2003iu,Soma:2014tp}, sensor location~\cite{Krause:2008vo}, text summarization~\cite{Lin:2010wpa,Lin:2011wt}, and maximum entropy sampling~\cite{Lee:2006cm}.
Many of these problems can be formulated as non-negative monotone submodular function maximization under a cardinality constraint or a matroid constraint.
Namely,
\begin{align}
  \text{(Cardinality constraint)} & \quad \text{maximize }f(S)  \quad \text{subject to } |S|\leq K, \quad S\subseteq E. \label{eq:card_problem} \\
  \text{(Matroid constraint)} & \quad \text{maximize }f(S) \quad \text{subject to } S\in\mathcal{I}, \quad S\subseteq E, \label{eq:matroid_problem}
\end{align}
where $f\colon 2^E \to \bbR_+$ is a monotone submodular function, $K \in \bbZ_+$ is a non-negative integer, and $\mathcal{M}=(E, \mathcal{I})$ is a matroid with independent family $\mathcal{I}$.
Note that a matroid constraint includes a cardinality constraint as a special case: Choose the matroid in~\eqref{eq:matroid_problem} to be the uniform matroid of rank $K$.


In some applications mentioned before, the amount of input data is much larger than the main memory capacity of individual computers.
Then, it is natural to consider the \emph{streaming model}, where each item in the ground set $E$ arrives sequentially, and we are allowed to use a small amount of memory.
Unless stated otherwise, we always consider \emph{single-pass algorithms}, that is, algorithms that scan the entire stream only once.

Submodular maximization under the streaming model has received much attention recently. Algorithms with various approximation ratios and space requirements
have been proposed for the cardinality constraint~\cite{Badanidiyuru:2014ib,Kazemi2019}, the knapsack constraint~\cite{HuangKakimuraWADS2019,Huang2019,Yu:2016}, and the matroid
constraint~\cite{DBLP:journals/mp/ChakrabartiK15,DBLP:conf/icalp/ChekuriGQ15}.
However, there are only a few inapproximability results.
McGregor and Vu~\cite{McGregor2019} showed that any streaming algorithm for maximizing a coverage function under a cardinality constraint with approximation ratio better than $1-\frac{1}{e}$ requires $\Omega\left(\frac{n}{K^2}\right)$ space, where $n := |E|$ is the number of elements.
Norouzi-Fard~et~al.~\cite{NorouziFard:2018vb} showed that any streaming algorithm for maximizing a monotone submodular function under a cardinality constraint with approximation ratio better than $\frac{K}{2K-1}$ requires $\Omega\left(\frac{n}{K}\right)$ space, assuming that we can only evaluate function values of feasible sets, which we call the \emph{weak oracle model}.
A standard value oracle is called \emph{strong} for comparison.

\subsection{Our contributions}
\begin{table}[t!]
  \centering
  \caption{Summary of results}\label{tab:intro}
  \begin{tabular}{llrrrr}
  \toprule
  Constraint & & Approximation ratio & Space usage & Oracle &  Reference\\
  \midrule \midrule
  \multirow{5}{*}{Cardinality} & \multirow{2}{*}{Algorithm} & $\frac{1}{2}-\varepsilon$ & $O\left(\frac{K }{\varepsilon}\right)$& weak & \cite{Kazemi2019} \\
    & & $\frac{K}{2K-1}- \varepsilon $ & $\widetilde{O}\left(\frac{K2^{2K}}{\varepsilon}\right)$ & weak & Theorem~\ref{thm:alg-weak-oracle-cardinality} \\ \cline{2-6}
  & \multirow{4}{*}{Hardness}  & $1-{\left(1-\frac{1}{K}\right)}^K+\varepsilon, \forall \varepsilon > 0$ & $\Omega\left(\frac{n}{K^2}\right)$ & strong &  \cite{McGregor2019} \\
  & & $\frac{K}{2K-1}+\varepsilon, \forall\varepsilon > 0$& $\Omega\left(\frac{n}{K}\right)$ & weak & \cite{NorouziFard:2018vb} \\
  & & $\frac{2}{2+\sqrt{2}}+\varepsilon, \forall\varepsilon > 0$ & $\Omega\left(\frac{n}{K^2}\right)$ & strong & Theorem~\ref{thm:lower-bound-cardinality}  \\
  \midrule
  \multirow{4}{*}{Matroid}& \multirow{2}{*}{Algorithm}& $\frac{1}{4}$ & $K \log^{O(1)}n$& strong & \cite{DBLP:journals/mp/ChakrabartiK15,DBLP:conf/icalp/ChekuriGQ15} \\
  && $\frac{1}{2}-\varepsilon $ & $\widetilde{O}\left(\frac{K^{5K+1}}{\varepsilon} \right)$ & weak & Theorem~\ref{thm:alg-weak-oracle-matroid} \\ \cline{2-6}
  & Hardness & $\frac{K}{2K-1}+\varepsilon, \forall\varepsilon > 0$ & $\Omega\left(\frac{n}{K}\right)$ & strong & Theorem~\ref{thm:lower-bound-matroid}  \\
  \bottomrule
  \end{tabular}
\end{table}

The first contribution of this work is giving inapproximability results for cardinality and matroid constraints that beat $1-\frac{1}{e}$ in the strong oracle model for the first time.

Before explaining our results, we first note that, in the context of submodular maximization, it is standard to assume that a value oracle of a function $f$ is given and the complexity of algorithms is measured based on the number of oracle calls~\cite{Badanidiyuru:2014ib,DBLP:journals/mp/ChakrabartiK15,DBLP:conf/icalp/ChekuriGQ15,HuangKakimuraWADS2019,Huang2019,Kazemi2019,Yu:2016}.
However, a value oracle of $f$ is too powerful in the streaming setting if we are allowed to put an exponential number of queries.
In fact, if we have a free access to the value oracle, we can maximize $f$ even without seeing the stream by querying about every subset.
This observation leads to the following natural model, which we call the \emph{element-store model}.
\begin{definition}[Element-store model]
  Let $E = \{e_1,\ldots,e_n\}$ be the ground set and $f\colon 2^E \to \mathbb{R}_+$ be a set function.
  A streaming algorithm in the \emph{element-store model} maintains a set of elements $S$, which is initially an empty set, and possess an additional memory $M$.
  At step $t \in \{1,\ldots,n\}$, the item $e_t$ is given to the algorithm, and the algorithm updates $S$ and the content of $M$ using the values of $f(S')$ for $S' \subseteq S \cup \{e_t\}$ and the content of $M$.
  Finally, the algorithm outputs a subset of $S$.
  The space complexity of the algorithm is the sum of the number of words stored in $M$ and the maximum size of $S$ over the $n$ steps.
\end{definition}
The weak oracle model is equivalent to constraining $S$ always to be a feasible set.
We note that all known streaming algorithms for submodular function maximization~\cite{Badanidiyuru:2014ib,DBLP:journals/mp/ChakrabartiK15,DBLP:conf/icalp/ChekuriGQ15,HuangKakimuraWADS2019,Huang2019,Kazemi2019,Yu:2016} lie in the element-store model.
Now, we state our results.
\begin{theorem}\label{thm:lower-bound-cardinality}
  For any $K \in \bbN$ and $\varepsilon > 0$, any (randomized) streaming algorithm for monotone submodular function maximization under a cardinality constraint in the element-store model with approximation ratio $\frac{2}{2+\sqrt{2}}+\varepsilon\approx 0.585 +\varepsilon$ requires $\Omega\left(\frac{n}{K^2}\right)$ space.
\end{theorem}
\begin{theorem}\label{thm:lower-bound-matroid}
  For any $K \in \bbN$ and $\varepsilon > 0$, any (randomized) streaming algorithm for monotone submodular function maximization under a partition matroid constraint in the element-store model with approximation ratio $\frac{K}{2K-1} + \varepsilon$ requires $\Omega\left(\frac{n}{K}\right)$ space.
\end{theorem}
Indeed, the same inapproximability results hold for any streaming algorithm in the element-store model with unbounded computational power and memory space for $M$, as long as the number of elements stored in $S$ is bounded.
The proof techniques can be found in Section~\ref{subsec:technique}.

Next, we complement the previous and obtained inapproximability results by showing (weak-oracle) streaming algorithms for cardinality and matroid constraints.
We first present a weak-oracle streaming algorithm for a cardinality constraint with approximation ratio $\frac{K}{2K-1}-\varepsilon$, which is slightly better than the previous best approximation ratio of $\frac{1}{2}$~\cite{Badanidiyuru:2014ib,Kazemi2019} and exactly matches the known inapproximability for the weak oracle model~\cite{NorouziFard:2018vb}.
Although the space usage is exponential in $K$, it does not depend on the number of elements $n$.
\begin{theorem}\label{thm:alg-weak-oracle-cardinality}
  There exists a weak-oracle $\left(\frac{K}{2K-1}-\varepsilon\right)$-approximation streaming algorithm for monotone submodular function maximization under a cardinality constraint with $O\Bigl(\frac{K 2^{2K}\log(K/\varepsilon)}{\varepsilon}\Bigr)$ space.
\end{theorem}
Then, we extend the algorithm given in Theorem~\ref{thm:alg-weak-oracle-cardinality} to a weak-oracle streaming algorithm for a matroid constraint with approximation ratio $\frac{1}{2}-\varepsilon$, which almost matches the inapproximability of $\frac{K}{2K-1}+\varepsilon$ given in Theorem~\ref{thm:lower-bound-matroid}.
This almost settles the the approximation ratio for a matroid constraint that can be achieved by a streaming algorithm with space complexity independent of $n$, for both the weak and strong oracle models.
\begin{theorem}\label{thm:alg-weak-oracle-matroid}
  There exists a weak-oracle $\left(\frac{1}{2} -\varepsilon\right)$-approximation streaming algorithm for monotone submodular function maximization under a matroid constraint with $O\left(\frac{K^{5K+1} \log(K/\varepsilon)}{\varepsilon}\right)$ space.
\end{theorem}

All the previous and obtained results are summarized in Table~\ref{tab:intro}. Here, $\widetilde{O}(\cdot)$ hides a polylogarithmic factor in $\frac{K}{\varepsilon}$.

\subsection{Our techniques}\label{subsec:technique}


\paragraph{Lower bound construction}
We first describe the intuition behind our proof of Theorem~\ref{thm:lower-bound-cardinality}.
The elements of the ground set $E$ are colored either blue, red, or purple, and we have a large number of $n-K$ blue elements, $K-1$ red elements, and one purple element.
Let $B$, $R$, and $P$ be the set of blue, red, and purple elements, respectively, that is, $|B|\;=n-K$, $|R|\;= K-1$, and $|P|\;=1$.
We note that the colors of elements are not revealed to algorithms.
We design our monotone submodular function $f\colon2^E \to \bbZ_+$ so that it is \emph{colorwise-symmetric} meaning that the value of $f(S)$ is uniquely determined by the number of blue, red, and purple elements in a subset $S$.
We write $f(b,r,p)$ to denote the value of $f(S)$ when there are $b$ blue elements, $r$ red elements, and $p$ purple elements in $S$.
We will assume that $f(0, K-1, 1)$ gives the optimal value.

In the input stream, blue and red elements arrive in a random order, and then the purple element arrives at the end.
We design $f$ so that it is hard to distinguish blue and red elements~(without using the purple element).
More precisely, $f$ satisfies the property $f(b+1, 0, 0)=f(b, 1, 0)$ for any non-negative integer $b$.
As the number of blue elements is much larger than that of red elements and the space is limited, with high probability, we must immediately throw away red elements from the memory right after they arrive.
Thus, with high probability, we obtain the values $f(K,0,0)$, $f(K-1,1,0)$ or $f(K-1,0,1)$, i.e., most of the time the algorithm ends up with at least $K-1$ blue elements.
On the basis of some ideas suggested by computer simulations, we construct $f$ so that the maximum of the three values is small.

The proof outline of Theorem~\ref{thm:lower-bound-matroid} is similar, but is more involved.
We first regard that the ground set $E$ is partitioned into classes $C_1,\ldots,C_K$ such that $|C_1| = \cdots = |C_{K-1}|=:m$ and $|C_K| = 1$, and we constrain that the output set takes at most one element from each class, which is a partition matroid constraint.
For each class $i$, there is a unique ``right'' element, referred to as the red element of the class, and for the first $K-1$ classes, there are a large number of ``wrong'' elements, referred to as blue elements of the class.
We will define a monotone submodular function $f\colon 2^E \to \bbR_+$ whose value is determined by (1) the presence/absence of the red element of class from 1 to $K$, and (2) the number of blue elements of class from 1 to $K-1$.
More precisely, given a set $S$, we denote by $r_i$ and $b_i$ the numbers of red and blue elements of class $i$ in $S$ for $1 \leq i \leq K$, respectively, and then $f(S)$ takes the form of $f(r_1, r_2,\ldots, r_K; b_1, b_2,\ldots, b_K)$.
We note that $r_i \in \{0,1\}$ for all $i$, and $b_K$ should always be $0$.
We call such a function \emph{colorwise-symmetric with respect to the partition $\{C_1,\ldots,C_K\}$.}
We will assume that $f(1,\ldots,1; 0,\ldots,0)$ gives the optimal value.

In the input stream, for each $i \in \{1,\ldots,K-1\}$ in this order, the blue and red elements of class $i$ arrive in a random order, and then the unique red element of class $K$ arrives.
We design $f$ so that it is hard to distinguish blue and red elements in each class.
More precisely, $f$ satisfies the property
\begin{align*}
& f(r_1,\ldots, r_{i-1}, 1,0,\ldots,0; b_1,\ldots, b_{i-1}, b_{i}\phantom{+1},0,\ldots,0) \\
= & f(r_1,\ldots, r_{i-1}, 0,0,\ldots,0; b_1,\ldots, b_{i-1}, b_{i}+1,0,\ldots,0)
\end{align*}
for any $1\leq i \leq K-1$, $r_1,\ldots,r_{i-1} \in \{0,1\}$, and $b_1,\ldots,b_i \in \{0,1,\ldots,m\}$.
Combined with the monotonicity of $f$, we can show that  the maximum value we can obtain via any algorithm is $f(0,\ldots, 0, 1; 1,\ldots,1,0)$ with high probability.
Again on the basis of some ideas suggested by computer simulations, we can construct $f$ so that this value is small.

We note that we took a different approach from the information-theoretic argument based on communication complexity used to show existing lower bounds~\cite{McGregor2019,NorouziFard:2018vb}, because we wanted to show lower bounds when the value oracle for a submodular function is available, and it is not clear how we can integrate it in the communication complexity setting.
In~\cite{McGregor2019}, coverage functions were explicitly constructed from instances of a communication complexity problem, and hence we can regard that the sets used to define the coverage functions are given one by one in a streaming fashion, and we do not need the value oracle.
In~\cite{NorouziFard:2018vb}, the issue was avoided by assuming that the value oracle is weak.

\paragraph{Our algorithms}
Our algorithms for cardinality and matroid constraints, given in Theorems~\ref{thm:alg-weak-oracle-cardinality} and~\ref{thm:alg-weak-oracle-matroid}, all use branching, depending on the
property of the first element $o_1$ of the optimal solution $\opt$ in the stream.
Here we explain the simplest case of cardinality constraint to highlight the basic ideas.
We devise a general procedure which takes two parameters $k$ and $s$.
The former is the upper bound on the size of the optimal solution while the latter is the allowed size of the solution.
Such a procedure would guarantee that the returned solution achieves the approximation ratio of $\frac{s}{k+s-1}$, where we observe that the ratio improves when $s$ is large relatively to $k$.

In the first branch, we assume that the value of $o_1$ is sufficiently large and we simply take the first element $e$ whose value is above a certain threshold and then recurse on all the elements after $e$ (with parameters $k-1$ and $s-1$).
Doing this guarantees that the element we first take is of large value (``bang for the buck''),
and more critically, $o_1$ (and hence the rest of $\opt$) is not ``missed'' in the recursion.
In the second branch, we assume that the value of $o_1$ is too small and we can as well just focus on $\opt-o_1$, by recursing directly on all remaining elements (with parameters $k-1$ and $s$).
Even though the value of $\opt-o_1$
is slightly smaller than $\opt$, the approximation ratio for the recursion is improved, as the available space $s$ grows relatively to the optimal solution size.

For the case of matroid constraint, the above branching strategy need more careful handling.
It is based on the idea of taking the first element that \emph{resembles} $o_1$ and use it to recurse on $\opt-o_1$.
There is an extra issue that the element $e$ resembling $o_1$ may not
form an independent set together with $\opt-o_1$. This issue is circumvented by using an extra set of candidates
of $o_1$, based on a known fact in matroid theory.

\subsection{Related work}

Maximizing a monotone submodular function subject to various constraints is a subject that has been extensively studied in the literature.
Although the problem is NP-hard even for a cardinality constraint, it can be approximated in polynomial time within a factor of $1-\frac{1}{e}$.
See e.g.,~\cite{Badanidiyuru:2013jc,FNS_cardinality,FisherNemhauserWolsey,Wolsey:1982}.
On the other hand, even for a cardinality constraint, we need an exponential number of function evaluations to obtain approximation ratio better than $1-\frac{1}{e}$~\cite{Nemhauser:1978dm,Vondrak:2013ia}.
Also, even when the submodular function is explicitly given (as a coverage function), Feige~\cite{Feige:1998gx} proved that the problem with a cardinality constraint cannot be approximated in polynomial time within a factor of $1-\frac{1}{e}+\varepsilon$ for any constant $\varepsilon >0$ unless P is equal to NP\@.
Besides a cardinality constraint, the problem has also been studied under
(multiple) matroid constraint(s), $p$-system constraint, multiple knapsack constraints.
See~\cite{Calinescu:2011ju,ChanSODA2017,Chan2017,DBLP:journals/siamcomp/ChekuriVZ14,ene_et_al:LIPIcs:2019:10629,ene_et_al:LIPIcs:2019:10630,Filmus:2014,Kulik:2013ix,Lee:2010,yoshida_2018} and the references therein.

\textit{Multi-pass streaming algorithms}, where we are allowed to read a stream of the input multiple times, have also been studied~\cite{Badanidiyuru:2013jc,DBLP:journals/mp/ChakrabartiK15,HuangKakimuraMultiPass2018,Huang2019}.
In particular, Chakrabarti and Kale~\cite{DBLP:journals/mp/ChakrabartiK15} gave an $O(\varepsilon^{-3})$-pass streaming algorithms for a generalization of the maximum matching problem and the submodular maximization problem with cardinality constraint.
Huang and Kakimura~\cite{HuangKakimuraMultiPass2018} designed an $O(\varepsilon^{-1})$-pass streaming algorithm with approximation guarantee $1/2-\varepsilon$ for the knapsack-constrained problem.
Other than the streaming setting, recent applications of submodular function maximization to large data sets have motivated new directions of research on other computational models including parallel computation model such as the MapReduce model~\cite{Barbosa2016,Barbosa2015,Kumar:2015} and the adaptivity analysis~\cite{BalkanskiRS19,Balkanski2018,ChekuriQ19,EneN19}.

The maximum coverage problem is a special case of monotone submodular maximization under a cardinality constraint where the function is a set-covering function.
For the special case, McGregor and Vu~\cite{McGregor2019} and Batani~et~al.~\cite{Bateni:2017} gave a $(1-e^{-1}-\varepsilon)$-approximation algorithm in the multi-pass streaming setting.

\subsection{Organization}
We prove our lower bound for strong-oracle algorithms for a cardinality constraint (Theorem~\ref{thm:lower-bound-cardinality}) and a matroid constraint (Theorem~\ref{thm:lower-bound-matroid}) in Sections~\ref{sec:cardinality} and~\ref{sec:matroid}, respectively.
We explain our weak-oracle algorithms and analyze them (Theorems~\ref{thm:alg-weak-oracle-cardinality} and~\ref{thm:alg-weak-oracle-matroid}) in Section~\ref{sec:algo}.


\section{Lower Bounds for Cardinality Constraints}\label{sec:cardinality}

In this section, we prove Theorem~\ref{thm:lower-bound-cardinality}.
As described in the introduction,
the ground set $E$ is partitioned into a blue set $B$, a red set $R$, and a purple set $P$, where $|B|\;=n-K$, $|R|\;= K-1$, and $|P|\;=1$.
We design a colorwise-symmetric function $f\colon2^E \to \bbZ_+$ such that $f(b+1, 0, 0)=f(b, 1, 0)$ for any non-negative integer $b\leq |B|-1$ and the values $f(K,0,0)$, $f(K-1,1,0)$ and $f(K-1,0,1)$ are small.
More specifically, we show the following:

\begin{lemma}\label{lem:hard-function-cardinality}
  For any large enough integer $n$ and any integer $h \geq K$, there exists a colorwise-symmetric function $f\colon2^E \to \bbZ_+$ with $|E| = n$ that satisfies the following conditions.
  \begin{enumerate}
    \itemsep=0pt
    \item[(i)] $f$ is monotone submodular.
    \item[(ii)] [Indistinguishability] $f(b+1, 0,0) = f(b,1,0)$ holds for all  $0 \leq b \leq n-K-1$.
    \item[(iii)] [Output value] $f(K, 0, 0) = f(K-1, 1, 0) = hK +  \frac{(K-1)K}{2}$ and $f(K-1, 0,1)  = {(K-1)}^2 +\frac{h(h+1)}{2}$ hold.
    \item[(iv)] [Optimal value] $f(0, K-1, 1) = (K-1)(h+K-1) +\frac{h(h+1)}{2}$ holds.
  \end{enumerate}
\end{lemma}

We defer the construction of our hard function and its analysis to Sections~\ref{subsec:hard-function-cardinality-construction}--\ref{subsec:hard-function-cardinality-proof}.

Below we prove Theorem~\ref{thm:lower-bound-cardinality} using Lemma~\ref{lem:hard-function-cardinality}.
We will use the following bound in the proof.
\begin{proposition}\label{pro:binomial}
  We have
  \[
    \left( 1 - \frac{k^2}{n}\right)  \frac{n^k}{k!} \leq \binom{n}{k} \leq \frac{n^k}{k!}.
  \]
\end{proposition}
\begin{proof}
  The claim holds from
  \[
    \binom{n}{k}=\frac{\prod_{i=0}^{k-1}(n-i)}{k!}
  \]
  and
  \[
    n^k \geq \prod_{i=0}^{k-1} (n-i) \geq {(n-k)}^k = {\left(1-\frac{k}{n}\right)}^k n^k \geq \left(1-\frac{k^2}{n}\right)n^k.
    \qedhere
  \]
\end{proof}

\begin{proof}[Proof of Theorem~\ref{thm:lower-bound-cardinality}]
  Let $h \geq K$ be an integer determined later, and let $f\colon 2^E \to \bbZ_+$ with $|E|=n$ be the colorwise-symmetric function as in Lemma~\ref{lem:hard-function-cardinality}.

  Let $\mathcal{D}$ be the uniform distribution over orderings $(e_1,\ldots,e_n)$ of elements of $E$, conditioned on that $e_1,\ldots,e_{n-1}$ include all the red and blue elements.
  Note that $e_n$ is the (unique) purple element.
  By Yao's minimax principle, to prove Theorem~\ref{thm:lower-bound-cardinality}, it suffices to show that any deterministic streaming algorithm $A$ with $o\left(\frac{n}{K^2}\right)$ space on an input sampled from $\mathcal{D}$ does not achieve approximation ratio more than $\frac{K}{2K-1}$ in expectation.

  Let $(e_1,\ldots,e_n)$ denote a sequence of elements sampled from $\mathcal{D}$.
  Let $S_t$ be the set of elements that $A$ holds after the $t$-th step, that is, the $t$-th element $e_t$ has arrived and $A$ has updated the set of elements it holds (by adding $e_t$ and/or discarding elements already in the set).
  We define $S_0 = \emptyset$ for convenience, and note that an algorithm chooses a subset of $S_n$ as the output of $A$.
  Note also that, for each $1 \leq t \leq n$, the set $S_{t}$ is completely determined by $S_{t-1}$ and the values of $f(S)$ and $f(S \cup \{e_t\})\;(S \subseteq S_{t-1})$, as $A$ is deterministic.

  For a set of indices $I \subseteq \set{1,\ldots,n}$, we define $S_I = \{e_i \mid i\in I \}$.
  Then, for $t \in \{0,1,\ldots,n-1\}$, we iteratively define a \emph{canonical set} $I^*_t$ of indices (not elements)  as follows.
  First, we set $I^*_0 = \emptyset$.
  Then, for each $1 \leq t \leq n-1$, we define $I^*_{t}$ as the set of indices of elements in $S_t$ when $A$ had $S_{I^*_{t-1}}$ after the $(t-1)$-th step and all but at most one element in $S_{I^*_{t-1}} \cup \{e_t\}$ are blue.
  Note that $I^*_t$ is uniquely determined because $A$ is deterministic, and by Property~(ii) of Lemma~\ref{lem:hard-function-cardinality}, the value of $f(S_I \cup \{e_t\})$ for $I \subseteq I^*_{t-1}$ is uniquely determined from the size of $I$.

  We say that $A$ followed the \emph{canonical process} if $A$ holds the set $S_{I^*_{t}}$ after the $t$-th step for each $1 \leq t \leq n-1$.
  For $1 \leq t \leq n-1$, let $X_t$ be the event that $S_{I^*_{t-1}} $ has one or more red elements and $e_t$ is red.
  Then, the probability that $A$ does not follow the canonical process is bounded by the probability that $\bigvee_{t=1}^{n-1} X_t$ happens.
  First, we have
  \[
    \Pr[X_t] \leq \frac{\sum_{r=1}^{K-2}\binom{s}{r} \binom{n-s-2}{K-2-r}}{\binom{n-1}{K-1}},
  \]
  where $s$ is the space usage of the algorithm, because the probability that $S_{I^*_{t-1}}$ has $r$ red balls and $e_t$ is red is at most $\binom{s}{r} \binom{n-s-2}{K-2-r}/\binom{n-1}{K-1}$.
  Then by a union bound, we have
  \begin{align*}
    & \Pr\Bigl[\bigvee_{t=1}^{n-1} X_t \Bigr] \leq \sum_{t =1}^{n-1}\Pr[X_t] \leq (n-1) \frac{\sum_{r=1}^{K-2}\binom{s}{r} \binom{n-s-2}{K-2-r}}{\binom{n-1}{K-1}} \\
    & \leq (n-1) \frac{\sum_{r=1}^{K-2} \frac{s^r}{r!} \frac{{(n-s-2)}^{K-2-r}}{(K-2-r)!} }{\left(1-\frac{K^2}{n}\right) \frac{{(n-1)}^{K-1}}{(K-1)!}} \tag{By Proposition~\ref{pro:binomial}} \\
    & = \frac{K-1}{\left(1-\frac{K^2}{n}\right) {(n-1)}^{K-2}} \sum_{r=1}^{K-2} \binom{K-2}{r} s^{r} {(n-s-2)}^{K-2-r} \\
    & = \frac{K-1}{\left(1-\frac{K^2}{n}\right) {(n-1)}^{K-2}} \Bigl({(n-2)}^{K-2} - {(n-s-2)}^{K-2}\Bigr) \\
    & = \frac{K-1}{1-\frac{K^2}{n} } {\left(\frac{n-2}{n-1}\right)}^{K-2} \left(1 - {\left(1-\frac{s}{n-2}\right)}^{K-2}\right) \\
    & \leq \frac{K-1}{1-\frac{K^2}{n} }  \frac{s(K-2)}{n-2} \tag{By ${(1 - x)}^d \geq 1- dx$}\\
    & = \frac{1}{1-\frac{K^2}{n} }O\left(\frac{K^2s}{n}\right).
  \end{align*}
  Let $Y$ be the event that $S_{I^*_{n-1}}$ has one or more red elements.
  We have
  \begin{align*}
    & \Pr[Y]
    \leq 1 - \frac{\binom{n-s-1}{K-1}}{\binom{n-1}{K-1}}
    \leq 1 - \frac{\left(1-\frac{K^2}{n-s}\right)\frac{{(n-s-1)}^{K-1}}{(K-1)!}}{\frac{{(n-1)}^{K-1}}{(K-1)!}} \tag{By Proposition~\ref{pro:binomial}}  \\
    & = 1 - \left(1-\frac{K^2}{n-s}\right) {\Bigl(1-\frac{s}{n-1}\Bigr)}^{K-1} \\
    & \leq 1 - \left(1-\frac{K^2}{n-s}\right) \Bigl(1-\frac{s(K-1)}{n-1}\Bigr) \tag{By ${(1 - x)}^d \geq 1- dx$}\\
    & = \frac{K^2}{n-s} + \left(1-\frac{K^2}{n-s}\right)\frac{s(K-1)}{n-1}.
  \end{align*}

  As long as $K=o(\sqrt{n})$ and $s = o\left(\frac{n}{K^2}\right)$, the probability that none of $X_1,\ldots,X_{n-1}$, and $Y$ happens is at least $1-o(1)$ by setting the hidden constant in $s$ to be small enough.
  If none of the events has happened, the algorithm $A$ can only obtain values for sets $S$ with $S \subseteq S_{I^*_{t-1}} \cup \set{e_t}$ for some $1 \leq t \leq n$ and $|S|\; \leq K$.
  As $S_{I^*_{t-1}} \cup \set{e_t}$ for any $1 \leq t \leq n-1$ contains at most one red element and $S_{I^*_{n-1}}$ contains no red element, the value of $f(S)$ is upper-bounded by $\max\set{f(K,0,0),f(K-1,1,0),f(K-1,0,1)}$, which is given by Property~(iii) of Lemma~\ref{lem:hard-function-cardinality}.
  Recall that the optimal value is given by Property~(iv) of Lemma~\ref{lem:hard-function-cardinality}.
  Therefore, the approximation ratio (in expectation over $\mathcal{D}$) is
  \[
    (1-o(1)) \cdot \frac{\max\left\{hK +  \frac{(K-1)K}{2}, {(K-1)}^2 +\frac{h(h+1)}{2} \right\}}
    {(K-1)(h+K-1) +\frac{h(h+1)}{2}} + o(1) \cdot 1.
  \]
  The ratio is minimized when $h$ is $\lfloor\sqrt{2}(K-1)\rfloor$ or $\lceil\sqrt{2}(K-1)\rceil$.
  When $K$ approaches to infinity, the ratio is
  \[
    (1-o(1)) \cdot \frac{2}{2+\sqrt{2}} + o(1) > 0.585,
  \]
  as desired.
\end{proof}

\subsection{Construction of Hard Functions}\label{subsec:hard-function-cardinality-construction}

We first define our function, and then describe the intuition behind the construction.
The next two subsections give its correctness proof.

\begin{definition}\label{def:hard-function}
  We define a colorwise-symmetric function $f\colon 2^E \to \bbZ_+$ recursively by its marginal return:
  Define
  \[
   f(0,0,0)=0 \quad\text{and}\quad f(0,0,1)=\frac{h(h+1)}{2}.
  \]
  We denote the marginal returns of $f$ by
  \begin{align*}
  {\Dr}(b,r) & = f(b, r+1, 0) - f(b, r, 0) = f(b, r+1, 1) - f(b, r, 1),\\
  {\Db}(b,0,0) & = f(b+1, 0, 0) - f(b, 0, 0),\\
  {\Db}(b,0,1) & = f(b+1, 0, 1) - f(b, 0, 1),
  \end{align*}
  where they are defined by
  \begin{align*}
  \Dr (b, r) &=
  \begin{cases}
  K-1 + h - b & \text{if}\quad 0\leq b\leq h+r,\\
  K-1 - \left\lceil \frac{r+b-h}{2}\right\rceil & \text{if}\quad h+r+1\leq b\leq h+2(K-2)-r,\\
  0 & \text{if}\quad h+2(K-2)-r+1\leq b,
  \end{cases}\\
  \Db (b, 0, 0) & (= \Dr(b, 0))
  =
  \begin{cases}
  K-1 + h - b & \text{if}\quad 0\leq b\leq h,\\
  K-1 - \left\lceil \frac{b-h}{2}\right\rceil & \text{if}\quad h+1\leq b\leq h+2(K-2),\\
  0 & \text{if}\quad h+2(K-2)+1\leq b,
  \end{cases}\\
  \Db (b, 0, 1) &=
  \begin{cases}
  K-1  & \text{if}\quad 0\leq b\leq h,\\
  K-1 - \left\lceil \frac{b-h}{2}\right\rceil & \text{if}\quad h+1\leq b\leq h+2(K-2),\\
  0 & \text{if}\quad h+2(K-2)+1\leq b.
  \end{cases}
  \end{align*}
\end{definition}

The value $f(b,r,p)$ is determined in the following way: we start with
the ``base value'', $f(0,0,p)$ (the presence of the purple element), then add the blue elements one by one
until there are $b$ of them (each increasing the marginal value by $\Db(i,0,p)$ for $0 \leq i \leq b-1$), and then add the red elements one by one
until there are $r$ of them (each increasing the marginal value by $\Dr(b,i)$ for $0 \leq i \leq r-1$). In other words, we have

\begin{equation}\label{eq:f}
  f(b,r,p) =f(0,0,p) + \sum_{j=0}^{b-1}\Db(j, 0, p) +  \sum_{i=0}^{r-1} {\Dr}(b,i).
\end{equation}

\subsubsection{Ideas behind the function}\label{sec:idea}

We start with several simple observations. 
The marginal values $\Db(b,0,0)$, $\Db(b,0,1)$ and $\Dr(b,r)$ are \emph{monotonically non-increasing} as the parameters $b$ and $r$
increase\footnote{To see $\Dr(b,r)\geq \Dr(b+1,r)$, it suffices to show when $b=h+r$ and when $b=h+2(K-2)-r$.
They follow because $\Dr(h+r, r) = K-r-1 \geq K-1 - \lceil r + \frac{1}{2} \rceil = \Dr(h+r+1,r)$ and $\Dr(h+2(K-2)-r, r) = K-1- \lceil K-1 \rceil \geq 0 = \Dr(h+2(K-2)-r+1, r)$.
To see $\Dr(b,r)\geq \Dr(b,r+1)$, the only tricky cases are when $b=h+r+1$ and when $b=h+2(K-2)-r$.
It holds that $\Dr(h+r+1, r) = K-1-r-1 = \Dr(h+r+1,r+1)$ and $\Dr(h+2(K-2)-r, r) = 1\geq 0 =\Dr(h+2(K-2)-r, r+1)$.}.
This conforms with the property of diminishing marginal return of a submodular function.
The next thing to notice is that, by design, the values $\Db(b,0,0)$ are exactly identical to $\Dr(b,0)$. So $f(b+1,0,0) = f(b,0,0) + \Db(b,0,0)=
 f(b,0,0) + \Dr(b,0) = f(b,1,0)$ and this is required by Lemma~\ref{lem:hard-function-cardinality}(ii).

 We next explain why the marginal values $\Db(b,0,0)$, $\Db(b,0,1)$ and $\Dr(b,r)$ are chosen in such a manner.
 A concrete example can be very useful to reveal the patterns generated by them.
 Assume that $K=4$ (thus 3 red elements and 1 purple element, and a large number of blue elements) and we choose $h=4$.
 Table~\ref{tab:example} gives the function values when the purple element is absent or present,
 along with $\Dr(b,r)$, the marginal value of adding a red element. We observe that for every fixed $r$, $\Dr(b,r)$ is \emph{strictly} monotonically
 decreasing in $b$ until $b = h+r$ (we color it differently); when $b$ is beyond $h+r$, the same value $\Dr(b,r)$ appears twice before decreasing.
Notice that the smaller the $r=i$, $\Dr(b,i)$ decreases in $b$ more slowly (conforming with the submodularity).


\begin{table}[htbp]
  \caption{The function values along with the marginal value $\Dr (b,r)$ when $K=h=4$.
  After the value $b=h+r$ (we color it specially), the same $\Dr (b,r)$'s appear twice before decreasing.
  Also notice that the values $f(b+1,0,0) = f(b, 1,0)$ for all $0 \leq b \leq n-K-1$ (Lemma~\ref{lem:hard-function-cardinality}(ii)).}\label{tab:example}
  \begin{center}
    \begin{tabular}{c}
      \begin{minipage}{0.5\hsize}
  \centering
  \begin{tabular}{l|l|l|l|l|l|l|l|l}
    \hline
    \multicolumn{1}{|l|}{$p=0$} & \multicolumn{2}{l|}{$r=0$} & \multicolumn{2}{l|}{$r=1$} & \multicolumn{2}{l|}{$r=2$} & $r=3$   \\ \hline
    \multicolumn{1}{|l|}{$b$} & $f$ & $\Dr$ & $f$ & $\Dr$ & $f$ & $\Dr$ & $f$   \\ \hline \hline
    \multicolumn{1}{|l|}{$0$} &    $0$       &   $7$         &    $7$         &   $7$         &    $14$        &      $7$      &     $21$             \\ \hline
    \multicolumn{1}{|l|}{$1$} &    $7$       &   $6$         &      $13$      &   $6$        &    $19$       &   $6$         & $25$                       \\ \hline
    \multicolumn{1}{|l|}{$2$} &     $13$      &    $5$       &     $18$      &   $5$         &     $23$       &   $5$         &   $28$                     \\ \hline
    \multicolumn{1}{|l|}{$3$} &     $18$      &   $4$         &     $22$      &   $4$         &     $26$       &  $4$          &   $30$                     \\ \hline
    \multicolumn{1}{|l|}{$4$} &     $22$      &    \color{red}{$3$}     &    $25$        &     $3$      &   $28$        &  $3$          &   $31$                     \\ \hline
    \multicolumn{1}{|l|}{$5$} &       $25$    &    $2$          &     $27$       &  \color{red}{$2$ }       &    $29$        &  $2$          &     $31$                     \\ \hline
    \multicolumn{1}{|l|}{$6$} &    $27$       &   $2$          &     $29$      &     $1$        &   $30$         &  \color{red}{$1$ }         &   $31$                       \\ \hline
    \multicolumn{1}{|l|}{$7$} &    $29$        &  $1$              &     $30$      &    $1$         &   $31$         &   $0$         &   $31$                       \\ \hline
    \multicolumn{1}{|l|}{$8$} &     $30$       & $1$           &       $31$      &    $0$         &  $31$          &  $0$          &     $31$                     \\ \hline
    \multicolumn{1}{|l|}{$9$} &     $31$       &   $0$          &    $31$       &     $0$        &   $31$         &   $0$        &     $31$                  \\ \hline
    \multicolumn{1}{|l|}{$10$} &     $31$       &   ---       &    $31$       &    ---        &    $31$        &  ---         &     $31$      \\ \hline
  \end{tabular}
      \end{minipage}

      \begin{minipage}{0.5\hsize}
  \centering
  \label{tab:withPurple}
  \begin{tabular}{l|l|l|l|l|l|l|l|l}
    \hline
    \multicolumn{1}{|l|}{$p=1$} & \multicolumn{2}{l|}{$r=0$} & \multicolumn{2}{l|}{$r=1$} & \multicolumn{2}{l|}{$r=2$} & $r=3$   \\ \hline
    \multicolumn{1}{|l|}{$b$} & $f$ & $\Dr$ & $f$ & $\Dr$ & $f$ & $\Dr$ & $f$   \\ \hline \hline
    \multicolumn{1}{|l|}{$0$} &    $13$       &   $7$         &    $20$         &   $7$         &    $27$        &      $7$      &     $34$              \\ \hline
    \multicolumn{1}{|l|}{$1$} &    $16$       &   $6$         &      $22$      &   $6$        &    $28$       &   $6$         & $34$                       \\ \hline
    \multicolumn{1}{|l|}{$2$} &     $19$      &    $5$       &     $24$      &   $5$         &     $29$       &   $5$         &   $34$                     \\ \hline
    \multicolumn{1}{|l|}{$3$} &     $22$      &   $4$         &     $26$      &   $4$         &     $30$       &  $4$          &   $34$                     \\ \hline
    \multicolumn{1}{|l|}{$4$} &     $25$      &     \color{red}{$3$}     &    $28$        &     $3$      &   $31$        &  $3$          &   $34$                     \\ \hline
    \multicolumn{1}{|l|}{$5$} &       $28$    &    $2$          &     $30$       &    \color{red}{$2$}        &    $32$        &  $2$          &     $34$                     \\ \hline
    \multicolumn{1}{|l|}{$6$} &    $30$       &   $2$          &     $32$      &     $1$        &   $33$         &  \color{red}{$1$}       &   $34$                       \\ \hline
    \multicolumn{1}{|l|}{$7$} &    $32$        &  $1$              &     $33$      &    $1$         &   $34$         &   $0$         &   $34$                       \\ \hline
    \multicolumn{1}{|l|}{$8$} &     $33$       & $1$           &       $34$      &    $0$         &  $34$          &  $0$          &     $34$                     \\ \hline
    \multicolumn{1}{|l|}{$9$} &     $34$       &   $0$          &    $34$       &     $0$        &   $34$         &   $0$        &     $34$                  \\ \hline
    \multicolumn{1}{|l|}{$10$} &     $34$       &     ---       &    $34$       &    ---        &    $34$        &  ---         &     $34$      \\ \hline
  \end{tabular}
      \end{minipage}

    \end{tabular}
  \end{center}
\end{table}

\begin{table}[t!]
\end{table}

Indeed we choose all marginal values $\Dr(b,r)$, $\Db(b,0,0)$ and $\Db(b,0,1)$ with two objectives.
(1) The indistinguishability of red and blue elements (i.e., Lemma~\ref{lem:hard-function-cardinality}(ii)), which forces $\Dr(b,0)=\Db(b,0,0)$,
and (2) creating a particular pattern of $f(b,K-1,0)$ and $f(b,K-1,1)$. As can be observed in this example,
$f(b, K-1,1)$ is simply the optimal value for all $b$, while $f(b, K-1,0)$ increases
by the amount of $h$, $h-1$, down to 1 when $b$ increases from $0$ to $h$ and stop growing
when $b$ is beyond $h$. The fact that we want $f(b, K-1,0)$ to stop growing when $b$ is beyond $h$
explains why $\Dr(b,r)$ behaves somehow differently when $b$ is beyond $h+r$.

\subsection{Correctness of Function Values}\label{subsec:hard-function-cardinality-fvalue}

In this section, we show that the function $f$ defined in Definition~\ref{def:hard-function} satisfies Lemma~\ref{lem:hard-function-cardinality}~(ii)--(iv).

Since $\Db (b, 0, 0) = \Dr(b, 0)$ for each $b$, we immediately have~(ii) of Lemma~\ref{lem:hard-function-cardinality}.
Moreover, it follows from~\eqref{eq:f} that
\begin{align*}
  f(K, 0, 0) &= \sum_{j=0}^{K-1}\Db(j, 0, 0) = hK +  \frac{(K-1)K}{2},\\
  f(K-1, 1, 0) & = \Dr (K-1, 0) + \sum_{j=0}^{K-2}\Db(j, 0, 0) = \sum_{j=0}^{K-1}\Db(j, 0, 0) = hK +  \frac{(K-1)K}{2},\\
  f(K-1, 0,1) & = \sum_{j=0}^{K-2}\Db(j, 0, 1) + f(0,0,1) = (K-1)(K-1) +\frac{h(h+1)}{2}.
\end{align*}
It also holds that
\begin{align*}
  f(0, K-1, 1) = \sum_{i=0}^{K-2}\Dr(0, i) + f(0,0,1) = (K-1)(h+K-1) +\frac{h(h+1)}{2}.
\end{align*}
Thus (iii) and (iv) of Lemma~\ref{lem:hard-function-cardinality} follow.

\subsection{Monotonicity and Submodularity}\label{subsec:hard-function-cardinality-proof}
Below we prove Lemma~\ref{lem:hard-function-cardinality}~(i), that is, the function $f$ defined in Definition~\ref{def:hard-function} is monotone and submodular.
To this end, it suffices to show the following lemma.
\begin{lemma}\label{lem:subm}
  Let $f\colon 2^E \to \bbZ_+$ be the colorwise-symmetric function defined in Definition~\ref{def:hard-function}.
  Then we have
  \begin{enumerate}
  \itemsep=0pt
  \item $f(b_1,r_1, 1)- f(b_1,r_1, 0) \geq f(b_2,r_2, 1)- f(b_2,r_2, 0)\geq 0$ for $0\leq b_1 \leq b_2 \leq |B|$ and $0 \leq r_1 \leq r_2 \leq |R|$.
  \item $f(b_1,r_1+1, p_1)- f(b_1,r_1, p_1) \geq f(b_2,r_2+1, p_2)- f(b_2,r_2, p_2)\geq 0$ for $0\leq b_1 \leq b_2 \leq |B|$, $0 \leq r_1 \leq r_2 \leq |R|-1$, and $0 \leq p_1 \leq p_2 \leq 1$.
  \item $f(b_1+1,r_1, p_1)- f(b_1,r_1, p_1) \geq f(b_2+1,r_2, p_2)- f(b_2,r_2, p_2)\geq 0$ for $0\leq b_1 \leq b_2 \leq |B|-1$, $0 \leq r_1 \leq r_2 \leq |R|$, and $0 \leq p_1 \leq p_2 \leq 1$.
  \end{enumerate}
\end{lemma}
\begin{proof}
  \textbf{(1)}
  By definition, it holds that, for $\ell=1,2$,
  \[
    f(b_\ell,r_\ell, 1)- f(b_\ell,r_\ell, 0) = \sum_{j=0}^{b_\ell-1}\Bigl(\Db(j, 0, 1) - \Db(j, 0, 0)\Bigr) + f(0,0,1) - f(0,0,0).
  \]
  Since $\Db (j, 0, 0) = \Db(j, 0, 1)$ if $j\geq h+1$, the RHS is equal to
  \[
    \sum_{j=0}^{\min\set{b_\ell-1, h}}\Bigl(K-1 - (K-1+h-j) \Bigr) + \frac{h(h+1)}{2}
    \geq -\sum_{j=0}^{h}(h-j)+  \frac{h(h+1)}{2} = 0.
  \]
  Hence the marginal return with respect to $p$ is non-negative.
  Moreover,
  \begin{align*}
  \Bigl(f(b_1,r_1, 1)- f(b_1,r_1, 0) \Bigr) - \Bigl(f(b_2,r_2, 1)- f(b_2,r_2, 0) \Bigr)
    &= \sum_{j=\min\set{h, b_1-1}}^{\min\set{h, b_2-1}}(h-j) \geq 0,
  \end{align*}
  since $h-j\geq 0$ for $j\leq h$.
  Thus (1) holds.

  \noindent
  \textbf{(2)}
  We observe that $\Dr(b, r)$ is a monotonically non-increasing function with respect to $b$, and a monotonically non-increasing function with respect to $r$~(See the footnote in Section~\ref{sec:idea}).
  Hence it holds that
  \[
    \Dr(b_1, r_1) \geq \Dr (b_1, r_2) \geq \Dr (b_2, r_2)\geq 0.
  \]
  This implies that
  \[
    f(b_1,r_1+1, p_1)- f(b_1,r_1, p_1) = \Dr (b_1, r_1) \geq \Dr (b_2, r_2) = f(b_2,r_2+1, p_2)- f(b_2,r_2, p_2)\geq 0.
  \]
  Thus (2) holds.

  \noindent
  \textbf{(3)}
  We first observe the following claims.

  \begin{claim}\label{clm:Dr}
  Let $0 \leq r \leq |R|$.
  \begin{enumerate}
    \item For $0\leq b \leq h+ r-1$, it holds that $2\Dr(b+1, r) = \Dr(b, r) + \Dr(b+2, r)$.
    \item For $h +r\leq b$, it holds that $\Dr(b, r) = \Dr(b+1, r-1)$.
  \end{enumerate}
  \end{claim}
  \begin{proof}
  They follow from the definition of $\Dr$.
  (1) clearly holds when $b+2\leq h+r$, because $\Dr$ is a linear function with respect to $b$ .
  When $b = h+r-1$, (1) also holds since $2\Dr(b+1, r) = 2(K-1-r)$, $\Dr(b, r)=K-r$, and $\Dr(b+2, r)=K-1-(r+1)$.
  For (2), the equality holds when $h +r+1\leq b$, since $\Dr$ depends on $b+r$.
  Moreover, when $h+r=b$, the equality holds since $\Dr(b, r) = \Dr(b+1, r-1)=K-1-r$.
  \end{proof}

  \begin{claim}\label{clm:subm2}
  For $0\leq b \leq |B|-2$, $0 \leq r \leq |R|$, and $0 \leq p \leq 1$, it holds that
  \[
    f(b+1,r, p)- f(b,r, p) \geq f(b+2,r, p)- f(b+1,r, p)\geq 0.
  \]
  \end{claim}
  \begin{proof}
    We first consider the case when $b\leq h+r-2$.
    We will show that the first inequality holds by induction on $r$.
    The inequality holds when $r=0$ as $\Db (b, 0, p)\geq \Db (b+1, 0, p)$.
    Suppose that $r>0$.
    Then we have
    \begin{align*}
      f(b+1,r, p)- f(b,r, p) &= \Bigl(f(b+1,r-1, p)+\Dr(b+1, r-1)\Bigr) - \Bigl(f(b,r-1, p)+\Dr(b, r-1)\Bigr)\\
      &\geq f(b+2,r-1, p)+\Dr(b+1, r-1) - f(b+1,r-1, p)-\Dr(b, r-1)\\
      & = f(b+2,r, p) - f(b+1,r, p) \\
      & \quad\quad +2\Dr(b+1, r-1) - \Dr(b, r-1)- \Dr(b+2, r-1),
    \end{align*}
    where the inequality holds by the induction hypothesis.
    Therefore, it follows from Claim~\ref{clm:Dr}~(1) that for $0\leq b \leq h+ r-2$,
    \[
      f(b+1,r, p)- f(b,r, p) \geq  f(b+2,r, p) - f(b+1,r, p).
    \]
    Thus, the claimed inequality holds when $b\leq h+r-2$.
    This implies that, for $0\leq b \leq h+ r-2$, we have $f(b+1,r, p)- f(b,r, p)\geq f(h+r+1, r, p)-f(h+r, r, p)$.

    Next suppose that $h+r-1\leq b$.
    We may assume that $r>0$, as the case when $r=0$ easily follows.
    By definition, it holds that
    \[
      f(b+1,r, p)- f(b,r, p) = \sum_{i=0}^{r-1}\Bigl(\Dr(b+1, i) - \Dr(b,i)\Bigr) + \Db(b,0,p).
    \]
    Since $\Dr(b, i) = \Dr(b+1, i-1)$ for $i=1,\dots, r-1$ by Claim~\ref{clm:Dr}~(2), we have
    \[
      f(b+1,r, p)- f(b,r, p) = \Dr(b+1, r-1) - \Dr(b,0) + \Db(b,0,p) = \Dr(b+1, r-1),
    \]
    where the last equality follows from $\Dr(b,0) = \Db(b,0,p)$ for $b\geq h$.
    Similarly, $f(b+2,r, p)- f(b+1,r, p) = \Dr(b+2, r-1)$.
    Hence
    \[
      f(b+1,r, p)- f(b,r, p) = \Dr(b+1, r-1) \geq \Dr(b+2, r-1) = f(b+2,r, p)- f(b+1,r, p).
    \]
    Thus the claimed inequality holds.

    Moreover, for $h+r-1\leq b \leq |B|-2$, we have $f(b+1,r, p)- f(b,r, p)\geq f(|B|, r, p)-f(|B|-1, r, p)=\Dr(|B|, r-1)\geq 0$.
    Thus the monotonicity also holds.
  \end{proof}

  \begin{claim}\label{clm:subm3}
  For $0\leq b \leq |B|-1$, $0 \leq r \leq |R|$, and $0 \leq p \leq 1$, it holds that
  \[
    f(b+1,r, 0)- f(b,r, 0) \geq f(b+1,r, 1)- f(b,r, 1).
  \]
  \end{claim}
  \begin{proof}
    We will show the claim by induction on $r$.
    The inequality holds when $r=0$ as $\Db (b, 0, 0) \geq \Db (b, 0, 1)$.
    Suppose that $r>0$.
    Then
    \begin{align*}
      f(b+1,r, 0)- f(b,r, 0) &= \left(f(b+1,r-1, 0)+\Dr(b+1, r-1)\right) - \left(f(b,r-1, 0)+\Dr(b, r-1)\right)\\
      &\geq f(b+1,r-1, 1)+\Dr(b+1, r-1) - f(b,r-1, 1)-\Dr(b, r-1)\\
      &= f(b+1,r, 1) - f(b,r, 1),
    \end{align*}
    where the inequality holds by the induction hypothesis.
    Thus the claim holds.
  \end{proof}

  It holds that,
  for $0\leq b \leq |B|-1$, $0 \leq r_1\leq r_2 \leq |R|$, and $0 \leq p \leq 1$,
  \begin{equation}\label{eq:subm4}
    f(b+1,r_1, p)- f(b,r_1, p) \geq f(b+1,r_2, p)- f(b,r_2, p),
  \end{equation}
  since $f(b+1,r_\ell, p)- f(b,r_\ell, p) = \sum_{i=0}^{r_\ell-1}\left(\Dr(b+1, i) - \Dr(b,i)\right) + \Db(b,0,p)$ for $\ell=1,2$.

  Therefore, applying Claims~\ref{clm:subm2} and~\ref{clm:subm3} with~\eqref{eq:subm4}, we have
  \begin{align*}
    f(b_1+1,r_1, p_1)- f(b_1,r_1, p_1) &\geq f(b_1+1,r_2, p_1)- f(b_1,r_2, p_1)\\
    &\geq f(b_2+1,r_2, p_1)- f(b_2,r_2, p_1) \geq f(b_2+1,r_2, p_2)- f(b_2,r_2, p_2).
  \end{align*}
  Since the monotonicity follows by Claim~\ref{clm:subm2},
  we complete the proof of (3).
\end{proof}

\section{Lower Bounds for Matroid Constraints}\label{sec:matroid}

In this section, we prove Theorem~\ref{thm:lower-bound-matroid}.
As described in the introduction, we assume that the ground set $E$ is partitioned into classes $C_1,\ldots,C_K$ such that $|C_1| = \cdots = |C_{K-1}|=: m$ and $|C_K| = 1$.
Note that the number of elements $n := |E|$ is $(K-1)m+1$.
We design a monotone submodular function $f\colon 2^E \to \mathbb{Z}_+$ that is colorwise-symmetric with respect to the partition $\{C_1,\ldots,C_K\}$ such that it is hard to distinguish blue and red elements in each class whereas we need to hit all the red elements to get the optimal value.
We will specify the exact values of $f$ in the next section.
Here, we summarize the critical properties of $f$ and use them to prove Theorem~\ref{thm:lower-bound-matroid}.

\begin{lemma}\label{lem:hard-function-partition}
  For any large enough integer $m$ and any positive integer $K$, there exists a colorwise-symmetric function $f\colon2^E \to \bbZ_+$ with respect to a partition $E = \bigcup_{i=1}^K C_i$ with $|C_1|=\cdots =|C_{K-1}| = m$ and $|C_K|$ = 1 such that
\begin{enumerate}
    \itemsep=0pt
    \item[(i)] $f$ is monotone submodular.
    \item[(ii)] [Optimal value] $f(1,\ldots,1; 0,\ldots,0) =(2K-1)!$.
    \item[(iii)] [Output value] $f(0,\ldots, 0, 1; 1,\ldots,1,0) = K(2K-2)!$.
     \item[(iv)] [Indistinguishability] For any $ 1\leq i \leq K-1$ and $r_1,\ldots,r_{i-1},b_1,\ldots,b_i$, we have
     \begin{align*}
     & f(r_1,\ldots, r_{i-1}, 1,0,\ldots,0; b_1,\ldots, b_{i-1}, b_{i}\phantom{+1},0,\ldots,0) \\
     = & f(r_1,\ldots, r_{i-1}, 0,0,\ldots,0; b_1,\ldots, b_{i-1}, b_{i}+1,0,\ldots,0).
     \end{align*}
  \end{enumerate}
\end{lemma}

\begin{proof}[Proof of Theorem~\ref{thm:lower-bound-matroid}]
  We modify the proof of Theorem~\ref{thm:lower-bound-cardinality}.

  Let $f\colon 2^E \to \mathbb{Z}_+$ be the function as in Lemma~\ref{lem:hard-function-partition}.
  For $1 \leq i \leq K -1$, let $T_i = \{(i-1)m+1,\ldots,im\}$ and let $T_K = \{n\}$.
  Let $\mathcal{D}$ be the uniform distribution over orderings $(e_1,\ldots,e_n)$ of elements of $E$, conditioned on that for each $1 \leq i \leq K$, the set $\{e_t \mid t \in T_i\}$ consists of all the elements of class $i$.
  By Yao's minimax principle, to prove Theorem~\ref{thm:lower-bound-matroid}, it suffices to show that any deterministic streaming algorithm $A$ with $o\left(\frac{n}{K}\right)$ space on an input sampled from $\mathcal{D}$ does not achieve approximation ratio more than $\frac{K}{2K-1}$ in expectation.

  We define $S_t$ for $t \in \{0,1,\ldots,n\}$ and $S_I$ for $I \subseteq \{1,\ldots,n\}$  as in the proof of Theorem~\ref{thm:lower-bound-cardinality}.
  For $t \in \{0,1,\ldots,n-1\}$, we iteratively define a \emph{canonical set} $I_t^*$ of indices (not elements) as follows.
  First, we set $I_0^* = \emptyset$.
  Then, for each $1 \leq t \leq n-1$, we define $I_t^*$ as the set of indices of elements in $S_t$ when $A$ had $S_{I^*_{t-1}}$ after the $(t - 1)$-th step and all but at most one element in $S_{I^*} \cup \{e_t\}$ are blue.
  Note that $I_t^*$ is uniquely determined because $A$ is deterministic, and by Property (iv) of Lemma~\ref{lem:hard-function-partition}, the value of $f(S_I \cup \{e_t\})$ for $I \subseteq I^*_t$ is uniquely determined from the sizes of $I \cap T_1,\ldots,I \cap T_{K-1}$.

  We say that $A$ followed the \emph{canonical process} if $A$ holds the set $S_{I^*_{t}}$ after the $t$-th step for each $1 \leq t \leq n-1$.
  For $1 \leq t \leq n-1$, let $X_t$ be the event that $S_{I^*_{t-1}} $ has one or more red elements and $e_t$ is red.
  Then, the probability that $A$ does not follow the canonical process is bounded by the probability that $\bigvee_{t=1}^{n-1} X_t$ happens.

  To bound $\Pr[X_t]$, we introduce some notations.
  Let $s$ be the space usage of the algorithm.
  For $i \in \{1,\ldots,K-1\}$, let $s_i = |I^*_{t-1} \cap T_i|$.
  For $t \in \{1,\ldots,n-1\}$, let $i_t \in \{1,\ldots,K-1\}$ be the class that the $t$-th element belongs to, that is, the unique integer $i$ with $t \in T_i$.
  Then for any $t \in \{1,\ldots,n-1\}$, we have
  \begin{align*}
    \Pr[X_t]
    & \leq \sum_{r=1}^{K-2}\max_{\substack{r_1,\ldots,r_{K-1} \in \{0,1\}: \\\sum_{i \neq i_t} r_i \leq r}}\frac{s_{i_t}}{m}\prod_{i \neq i_t}\frac{ s_i^{r_i} {(m-s_i)}^{1-r_i}}{m}
    \leq \sum_{r=1}^{K-2} {\left(\frac{s}{r}\right)}^r \frac{m^{K-2-r}}{m^{K-1}} \\
    & = \frac{1}{m}\sum_{r=1}^{K-2} {\left(\frac{s}{rm}\right)}^r
    \leq \frac{s}{m(m-s)} = O\left(\frac{s}{m^2}\right).
  \end{align*}
  Here, we regard $r$ as the number of red elements in $S_{I^*_{t-1}}$ and regard $r_1,\ldots,r_{K-1} \in \{0,1\}$ as the numbers of red elements in $S_{I^*_{t-1} \cap T_1},\ldots,S_{I^*_{t-1} \cap T_{K-1}}$, respectively.
  The first inequality holds because the probability that $e_t$ is red is $\frac{s_{i_t}}{m}$ and the probability that $S_{I^*_{t-1} \cap T_i}$ has $r_i$ red elements is $\frac{s_i^{r_i}{(m-s_i)}^{1-r_i}}{m}$.

  Now by a union bound, we have
  \begin{align*}
    & \Pr\left[\bigvee_{t=1}^{n-1} X_t \right] \leq \sum_{t =1}^{n-1}\Pr[X_t]
    = O\left(\frac{sn}{m^2}\right)
    = O\left(\frac{Ks}{m}\right).
  \end{align*}

  Let $Y$ be the event that $S_{I^*_{n-1}}$ has one or more red elements.
  Then, we have
  \begin{align*}
    & \Pr[Y]
    \leq 1 - \prod_{i=1}^{K-1} \frac{m-s_i}{m}
    \leq 1 - \frac{(m-s)m^{K-2}}{m^{K-1}}
    = \frac{s}{m}.
  \end{align*}

  As long as $Ks=o(n)$, the probability that none of $X_1,\ldots,X_{n-1}$, and $Y$ happens is at least $1-o(1)$ by setting the hidden constant in $s$ to be small enough.
  If none of the events has happened, the algorithm $A$ can only obtain values for sets $S$ with $S \subseteq S_{I^*_{t-1}} \cup \set{e_t}$ for  $1 \leq t \leq n$ and $|S| \leq K$.
  As $S_{I^*_{t-1}} \cup \set{e_t}$ for any $1 \leq t \leq n-1$ contains at most one red element and $S_{I^*_{n-1}}$ contains no red element, the values the algorithm can observe is at most $K(2K-2)!$ by Properties~(iii) and (iv) of Lemma~\ref{lem:hard-function-partition}.
  Recall that the optimal value is $(2K-1)!$ by Property~(ii) of Lemma~\ref{lem:hard-function-partition}.
  Therefore, the approximation ratio (in expectation over $\mathcal{D}$) is
  \[
    (1-o(1)) \cdot \frac{K(2K-2)!}
    {(2K-1)!} + o(1) \cdot 1 = (1-o(1))\frac{K}{2K-1} + o(1). \qedhere
  \]
\end{proof}

\subsection{Construction of the Hard Function}\label{sec:def_func_matroid}

We begin by describing some characteristics of the function we will define. Let $\hat{b}_i = \min \{ b_i, 2(K-i)\}$.
The function will be expressed as a polynomial of $r_i$ and $\hat{b}_i$ for all $1 \leq i \leq K$. In other words, the number of blue elements matter only up to a certain ceiling: For class $i$, if there are more than $2(K-i)$ blue elements in the class $i$, the function value
is the same as if there are exactly $2(K-i)$ blue elements. To be more precise, we decree that
\[
f(r_1, \cdots r_K; b_1, \ldots, b_K) =  f(r_1, \ldots, r_K; \hat{b}_1, \ldots, \hat{b}_K).
\]

We now define the function $f$ recursively.
For $t=1,\dots K$, let $f_t$ be a function on the last $t$ classes $C_{K-(t-1)},\dots, C_K$, that is, $f_t$ takes the form of
\[
f_t(r_{K-(t-1)},\dots, r_K; b_{K-(t-1)},\dots, b_K)
\]
Define
\[
f_1(r_K; b_K) = r_K,
\]
and assume the function $f_{t-1}$ is already defined for some $t \geq 2$.
We then define the function $f_{t}$.
For that purpose, we give some notation.
Let $m_t = (2t-1)!$.
As we will see later, the function value of $f_{t-1}$ is between $0$ and $m_{t-1}$.
Suppose that we are given $r_{K-(t-1)},r_{K-(t-2)},\ldots, r_{K}$ and $b_{K-(t-1)}, b_{K-(t-2)},\ldots, b_{K}$.
Define\footnote{Here we note that the terms $\delta_{t-1}$, $d_t$, and $s_t$ indeed depend on the value of $r_{K-(t-1)},r_{K-(t-2)},\ldots, r_{K}$ and $b_{K-(t-1)}, b_{K-(t-2)},\ldots, b_{K}$. However, we choose to avoid the cumbersome notation of associating the former with the latter.}
\[
\delta_{t-1} = m_{t-1} - f_{t-1}(r_{K-(t-2)},\ldots, r_{K}; \hat{b}_{K-(t-2)},\ldots, \hat{b}_{K}).
\]
We also define
\[
d_t = 2(t-1)-\hat{b}_{K-(t-1)} \mbox{\quad and \quad} s_t = 1 - r_{K-(t-1)}.
\]
The term $d_t$~(resp., $s_t$) is simply the gap between $\hat{b}_{K-(t-1)}$~(resp., $r_{K-(t-1)}$) and its potential maximum.
We remark that both of them are non-negative, and $\hat{b}_{K-(t-1)}+d_t=2(t-1)$ and $r_{K-(t-1)}+s_t=1$.
We can then express $f_t$ as follows:
\begin{align}\label{eq:ft}
f_t(r_{K-(t-1)},\ldots, r_{K}; b_{K-(t-1)},\ldots, b_{K}) & = m_t - a_t \cdot d_t,
\end{align}
where
\[
a_t = 2m_{t-1}s_t + \delta_{t-1}(d_t-1).
\]

The function $f_t$ is set up in such a way so that
$f_t (1,\dots, 1;0,\dots, 0)$ is exactly $m_t=(2t-1)!$ for any $t$.
As we will show, this maximizes $f_t$.
On the other hand, $f_t(0,\dots,0;0,\dots, 0)=0$, which is the minimum of $f_t$.
%
To see the significance of the term $a_t$, recall that the term $\delta_{t-1}$ encodes the difference between the maximum of $f_{t-1}$
 and the actual value attained by the given
$r_{K-(t-2)},\ldots, r_{K}$ and $\hat{b}_{K-(t-2)},\ldots, \hat{b}_{K}$ --- therefore always non-negative.
Then we can regard $a_t$ as a linear combination of $s_t$ and $d_t$~(both are decided by the number of red/blue elements of class $K-(t-1)$), where the coefficients are respectively $2m_{t-1}$ and $\delta_{t-1}$~(both are decided by the number of red/blue elements from later classes $K-(t-1)+1, \ldots, K$).
See Lemma~\ref{lem:matroidFunctionMonotonicity} for a more precise summary of the above
discussion.

\subsection{Concrete Example and Some Observations}\label{sec:MatroidObs}

We present a concrete example to highlight several interesting properties of the function constructed in Section~\ref{sec:def_func_matroid}, and to share some of our experiences in searching for such a function.
Let $K=3$. All the function values when $0 \leq b_1 \leq 2(K-1)=4$ and  $0 \leq b_2  \leq 2(K-2)=2$ are shown
in Tables~\ref{tab:withOutLastElement} and~\ref{tab:withLastElement}.
We remark that having more blue elements does not increase the value further, as mentioned in Section~\ref{sec:def_func_matroid}.

\begin{table}[h!]
\scriptsize
\begin{tabular}{|l|r|r|r|l|r|r|r|l|r|r|r|l|r|r|r|}
\hline
 & \multicolumn{3}{l|}{$r_1=0$, $r_2=0$} &  & \multicolumn{3}{l|}{$r_1=1$, $r_2=0$} &  & \multicolumn{3}{l|}{$r_1=0$, $r_2=1$} &  & \multicolumn{3}{l|}{$r_1=1$, $r_2=1$} \\ \hline
\backslashbox[14mm]{$\hat{b}_1$}{$\hat{b}_2$}&    0   &  1      & 2      &\backslashbox[14mm]{$\hat{b}_1$}{$\hat{b}_2$}  &  0     & 1      &  2     & \backslashbox[14mm]{$\hat{b}_1$}{$\hat{b}_2$} &  0     & 1      &2        & \backslashbox[14mm]{$\hat{b}_1$}{$\hat{b}_2$} & 0      &  1     &  2     \\ \hline
 0&   0    & 48      &  72     &  0 &    48   &  96     & 120      & 0 &  48   &   72    &  72     &  0&   96    & 120      &120      \\ \hline
 1&   48    &  72     &    84   &  1&  84     &  108     &  120     & 1 &  72     &  84     &  84     & 1 &   108    &   120    & 120      \\ \hline
 2&   84    &    92   &   96    & 2 &   108    &  116     & 120      & 2 &  92     &   96    &  96     & 2 &  116     &      120 &   120    \\ \hline
 3&    108   &   108    & 108      & 3 &   120    &   120    &  120     &3  &   108    & 108      &  108     &3  &  120     & 120      &   120    \\ \hline
 4&    120   &   120    &   120    &  4&   120    &  120     &  120     & 4 &   120    &  120     & 120      & 4 &   120    & 120      & 120      \\ \hline
\end{tabular}
\caption{When the unique element of $C_3$ is absent, i.e., $r_3=0$.}\label{tab:withOutLastElement}
\end{table}

\begin{table}[h!]
\scriptsize
\begin{tabular}{|l|r|r|r|l|r|r|r|l|r|r|r|l|r|r|r|}
\hline
 & \multicolumn{3}{l|}{$r_1=0$, $r_2=0$} &  & \multicolumn{3}{l|}{$r_1=1$, $r_2=0$} &  & \multicolumn{3}{l|}{$r_1=0$, $r_2=1$} &  & \multicolumn{3}{l|}{$r_1=1$, $r_2=1$} \\ \hline
\backslashbox[14mm]{$\hat{b}_1$}{$\hat{b}_2$}&    0   &  1      & 2      &\backslashbox[14mm]{$\hat{b}_1$}{$\hat{b}_2$}  &  0     & 1      &  2     & \backslashbox[14mm]{$\hat{b}_1$}{$\hat{b}_2$} &  0     & 1      &2        & \backslashbox[14mm]{$\hat{b}_1$}{$\hat{b}_2$} & 0      &  1     &  2     \\ \hline
 0&   24   & 48      &  72     &  0 &    72   &  96     & 120      & 0 &  72  &   72    &  72     &  0&   120    & 120      &120      \\ \hline
 1&   60   &  72     &    84   &  1&  96     &  108     &  120     & 1 &  84     &  84     &  84     & 1 &   120    &   120    & 120      \\ \hline
 2&   88    &    92   &   96    & 2 &   112    &  116     & 120      & 2 &  96     &   96    &  96     & 2 &  120     &      120 &   120    \\ \hline
 3&    108   &   108    & 108      & 3 &   120    &   120    &  120     &3  &   108    & 108      &  108     &3  &  120     & 120      &   120    \\ \hline
 4&    120   &   120    &   120    &  4&   120    &  120     &  120     & 4 &   120    &  120     & 120      & 4 &   120    & 120      & 120      \\ \hline
\end{tabular}
\caption{When the unique element of $C_3$ is present, i.e., $r_3=1$.}\label{tab:withLastElement}
\end{table}

We can first observe that, individually, the values of a single red/blue element of the first $K-1$ classes are all equal, and the value of the
unique red element in the last class is half of them. In the present example, an element in the first two classes has value 48 while an element in the last class has value 24.
In fact, we have the following observation for the constructed function $f$.
Note that $b_K$ should always be 0.

\begin{lemma}\label{lem:obs1}
The constructed function $f$ satisfies the following.
\begin{align*}
f(0,\dots, 0, 1, 0, \dots, 0; 0, \dots, 0)
&= \begin{cases}
(2K-2)! & \mbox{if $1$ is at the last class $K$}\\
2(2K-2)! & \mbox{otherwise}\\
\end{cases}\\
f(0, \dots, 0;0,\dots, 0, 1, 0, \dots, 0)&= 2(2K-2)!,
\end{align*}
\end{lemma}
\noindent
The proof will be given in Section~\ref{sec:ProofObservations}.

In the present example, the optimal value is $120$, which is reached by a set of $2(K-1)$ blue elements of class $1$, i.e, $f(0,0,0;4,0,0)=120$, though it is infeasible.
The optimal value is also obtained by a set of all the red elements, which is equal to the sum of the values of each red element.
In the example, we see that $f(1,1,1;0,0,0)=f(1,0,0;0,0,0)+f(0,1,0;0,0,0)+f(0,0,1;0,0,0)=48+48+24=120$.

\begin{lemma}\label{lem:obs2}
The constructed function $f$ satisfies that
\begin{align*}
f(1,\ldots,1; 0,\ldots,0) &= f(1,0,\dots, 0; 0, \dots, 0)+f(0,1,0,\dots, 0; 0, \dots, 0)+ \cdots \\ 
& \hspace{1em}\dots  +f(0,\dots,0,1; 0, \dots, 0) =(2K-1)!,\\
f(0,\dots, 0; 2(K-1),0, \dots, 0) &=(2K-1)!.
\end{align*}
\end{lemma}
\noindent
The proof will be given in Section~\ref{sec:ProofObservations}.

By the construction, the function $f(r_1,r_2,r_3; b_1, b_2, b_3)$ can be expressed as the following polynomial:
\begin{align}
f(r_1,r_2,r_3; b_1, b_2, b_3) & = m_3 - \left(2 m_2 s_3 + \left( 2m_1 s_2 + s_1 (d_2 - 1) \right)d_2 (d_3-1) \right) \cdot d_3\nonumber \\
& =  120 - \left(12 s_3 + \left( 2 s_2 + s_1 (d_2 - 1) \right)d_2 (d_3-1) \right) \cdot  d_3,\label{equ:exampleForK3}
\end{align}
where we recall that $m_t = (2t-1)!$, $d_t = 2(t-1)-\hat{b}_{K-(t-1)}$ and $s_t = 1 - r_{K-(t-1)}$ for $t=1,2,3$.
We can observe that $f$ is a polynomial in $s_i$'s and $d_i$'s, where, for each monomial, $s_i$'s have degree at most 1 and $d_i$'s have degree at most 2.
This implies that \emph{the discrete second derivative with respect to $b_i$'s is constant when $b_i$ is in $[0, 2(K-i)]$}.
Take the current example with $r_1=r_2=r_3=b_2=0$.
When $b_1$ increases from 0 to 4, the discrete first derivative is 48, 36, 24, 12 and the discrete second derivative is a constant, which is 12.
It can be verified that the same property holds for all columns and all rows in Tables~\ref{tab:withOutLastElement} and~\ref{tab:withLastElement}.
In fact, we noticed in computer-aided search that imposing the additional constraint of constant second derivative does not change the ratio between the optimal value and the value an algorithm finds, making the found function more structured and generalizable.
Based on computer-aided search imposing this additional constraint for small values of rank $K$, we found the recurrence in the previous section.

\subsection{Correctness of the Function}

%
%

We first show the monotonicity of the constructed function $f_t$, which implies Lemma~\ref{lem:hard-function-partition}~(i) and (ii) when $t=K$.

\begin{lemma}\label{lem:matroidFunctionMonotonicity}
For any $t=1,2,\dots, K$, the constructed function $f_t$ satisfies the following.
\begin{enumerate}



\item $f_t$ is monotone in $r_{K-(t-1)},\ldots, r_K$ and in $b_{K-(t-1)},\ldots, b_K$.

\item $f_t$ reaches the maximum, which is $m_t = (2t-1)!$, when $r_{K-(t-1)}= \cdots =r_K = 1$.

\item $f_t$ reaches the minimum, which is $0$, when $r_{K-(t-1)}=\cdots = r_K=b_{K-(t-1)}=\cdots = b_K=0$.

\end{enumerate}
\end{lemma}

\begin{proof} 
We prove by induction on $t$. When $t=1$, it is straightforward to verify the lemma.
Consider when $t \geq 2$.

For (1), suppose that $r^1_{K-(t-1)} \geq r^2_{K-(t-1)}$, $\cdots$, $r^1_K \geq r^2_K$ and $b^1_{K-(t-1)} \geq b^2_{K-(t-1)}$, $\cdots$, $b^1_K \geq b^2_K$.
For $\ell=1,2$, let $\delta^\ell_{t-1}$, $d^\ell_t$, $s^\ell_t$ and $f^\ell_t$ denote the realizations
 of $\delta_{t-1}$, $d_t$, $s_t$ and $f_t$, based on $\bigcup_{j=K-(t-1)}^K \{r^\ell_j, b^\ell_j\}$.
 By the induction hypothesis, $f_{t-1}$ is monotone, and hence we see $0 \leq \delta^1_{t-1} \leq \delta^2_{t-1}$.
  Furthermore, as $b^1_{K-(t-1)} \geq b^2_{K-(t-1)}$, we have $\hat{b}^1_{K-(t-1)} \geq  \hat{b}^2_{K-(t-1)}$, implying that
 $d^1_t \leq d^2_t$. Now it holds by~\eqref{eq:ft} that
 \begin{align*}
 f^1_t - f^2_t &= \delta^2_{t-1}(d^2_t-1)d^2_t -\delta^1_{t-1}(d^1_t-1)d^1_t + 2m_{t-1}\left(s^2_t d^2_t-s^1_t d^1_t\right)\\
 &\geq \delta^2_{t-1}(d^2_t-1)d^2_t -\delta^1_{t-1}(d^1_t-1)d^1_t \\
 &\geq \delta^1_{t-1}\left( (d^2_t-1)d^2_t -(d^1_t-1)d^1_t\right)\geq 0
\end{align*}
where the first inequality holds because $s^2_t \geq s^1_t$ and $d^2_t \geq d^1_t$, the second inequality is from $\delta^2_{t-1} \geq \delta^1_{t-1}$, and the last inequality follows from the fact that the function $x(x-1)$ is monotonically increasing for a non-negative integer $x$.
(1) is then proved.

For (2), when $r_{K-(t-1)}=r_{K-(t-2)}= \cdots= r_K = 1$, by the induction hypothesis, $f_{t-1}$ takes the maximum $m_{t-1}$, and hence $\delta_{t-1}=0$. Then $a_t=0$, since $s_t=0$ and $\delta_{t-1}=0$.
Hence the $f_t$-value in this case is equal to $m_t$.
In addition, it holds that $m_t$ is an upped bound of $f_t$, since  $a_t d_t$ is non-negative in~\eqref{eq:ft}.
Indeed, $a_t d_t=0$ when $d_t=0$, and when $d_t\geq 1$, $a_t$ is non-negative since $\delta_{t-1}$ is non-negative by the induction hypothesis.
Thus $m_t$ is the maximum value of $f_t$, which implies (2).

For (3), when $r_{K-(t-1)}=r_{K-(t-2)}=\cdots = r_K= b_{K-(t-2)}=\cdots = b_K=0$, the induction hypothesis gives that $\delta_{t-1} = m_{t-1}$.
Since $s_t=1$ and $d_t = 2(t-1)$ when $r_{K-(t-1)}=b_{K-(t-1)}=0$, it holds that
\[
a_t d_t = m_{t-1} (2 s_t + d_t -1)d_t = m_{t-1} (2t-1)(2t-2)= (2t-1)!,
\]
which is exactly $m_t$. Thus $f_t(0,\dots, 0;0,\dots, 0)=0$ by~\eqref{eq:ft}.
Since $f_t$ is monotone by (1), it is minimum.

Therefore, (1)--(3) hold.
\end{proof}

We calculate the function value obtained by an algorithm~(Lemma~\ref{lem:hard-function-partition}~(iii)).

\begin{lemma}\label{lem:matroidFunctionOutput}
For any $t=1,2,\dots, K$, the constructed function $f_t$ satisfies the following.
\[
f_t(0,\dots,0,1; 1,\dots,1,0) = t\cdot (2t-2)!.
\]
\end{lemma}
\begin{proof}
We prove by induction on $t$. When $t=1$, it is straightforward to verify the lemma.
By the induction hypothesis, $f_{t-1}(0,\ldots,0,1; 1,\ldots,1,0) = (t-1)\cdot (2t-4)!$. Then
$\delta_{t-1}= (t-2)\cdot (2t-4)!$.
As $s_t=1$ and $d_t=2t-3$, it holds by~\eqref{eq:ft} that
\[
f_t(0,\ldots,0,1; 1,\ldots,1,0) =
(2t-1)! - \left( 2 \cdot (2t-3)! + (2t-4)! (t-2) (2t-4) \right) (2t-3) = t\cdot (2t-2)!,
\]
where the last equality is easy to verify.
\end{proof}

We next deal with submodularity of the constructed function.
%

\begin{lemma}\label{lem:matroidFunctionSubmodularity}
For any $t=1,2,\dots, K$, the constructed function $f_t$ satisfies the following:
Suppose that $r^1_{K-(t-1)} \geq r^2_{K-(t-1)}, \ldots, r^1_K \geq r^2_K$ and $b^1_{K-(t-1)}\geq b^2_{K-(t-1)}, \ldots, b^1_K \geq b^2_K$.
Then
\begin{enumerate}
\item For any integer $i$ with $ K-(t-1)\leq i \leq K$ such that $r^1_i=0$,
it holds that
\begin{align}
&f_t(r^1_{K-(t-1)}, \ldots, r^1_i+1,\ldots, r^1_K;b^1_{K-(t-1)},\ldots, b^1_K)  - f_t(r^1_{K-(t-1)}, \ldots, r^1_i,\ldots, r^1_K;b^1_{K-(t-1)},\ldots, b^1_K)\nonumber\\
&\leq f_t(r^2_{K-(t-1)}, \ldots, r^2_i+1,\ldots, r^2_K;b^1_{K-(t-1)},\ldots, b^2_K)  - f_t(r^2_{K-(t-1)}, \ldots, r^2_i,\ldots, r^2_K;b^2_{K-(t-1)},\ldots, b^2_K).\label{eq:subm_matroid_r}
\end{align}
\item For any integer $i$ with $ K-(t-1)\leq i \leq K$,
it holds that
\begin{align}
&f_t(r^1_{K-(t-1)}, \ldots, r^1_K;b^1_{K-(t-1)},\ldots, b^1_i+1,\ldots, b^1_K)  - f_t(r^1_{K-(t-1)}, \ldots, r^1_K;b^1_{K-(t-1)},\ldots, b^1_i,\ldots, b^1_K)\nonumber\\
&\leq f_t(r^2_{K-(t-1)}, \ldots, r^2_K;b^\ell_{K-(t-1)},\ldots, b^2_i+1,\ldots, b^2_K)  - f_t(r^2_{K-(t-1)}, \ldots, r^2_K;b^2_{K-(t-1)},\ldots, b^2_i,\ldots, b^2_K).\label{eq:subm_matroid_b}
\end{align}
\end{enumerate}
\end{lemma}

\begin{proof} We proceed by induction on $t$. The base case $t=1$ is easy to verify. For the induction step where $t\geq 2$, let $d^\ell_t$, $s^\ell_t$ and $\delta^\ell_{t-1}$ denote the realizations of $d_t$, $s_t$, and $\delta_{t-1}$ based on $\bigcup_{j=K-(t-1)}^K \{r^\ell_j, b^\ell_j\}$ for $\ell=1,2$.
 As $b^1_{K-(t-1)} \geq b^2_{K-(t-1)}$, we have $\hat{b}^1_{K-(t-1)} \geq  \hat{b}^2_{K-(t-1)}$, and thus
 $d^1_t \leq d^2_t$.

First assume that $i\geq K-(t-2)$.
Then it follows from~\eqref{eq:ft} that the left-hand side of~\eqref{eq:subm_matroid_r} is equal to
\begin{align*}
d^1_t(d^1_t-1) & ( f_{t-1}(r^1_{K-(t-2)}, \ldots, r^1_i+1,\ldots, r^1_K;b^1_{K-(t-2)},\ldots, b^1_K)  \\
& - f_{t-1}(r^1_{K-(t-2)}, \ldots, r^1_i,\ldots, r^1_K;b^1_{K-(t-2)},\ldots, b^1_K) ).
\end{align*}
Similarly, the right-hand side of~\eqref{eq:subm_matroid_r} is equal to
\begin{align*}
d^2_t(d^2_t-1) & ( f_{t-1}(r^2_{K-(t-2)}, \ldots, r^2_i+1,\ldots, r^2_K;b^2_{K-(t-2)},\ldots, b^2_K) \\
& - f_{t-1}(r^2_{K-(t-2)}, \ldots, r^2_i,\ldots, r^2_K;b^2_{K-(t-2)},\ldots, b^2_K) ).
\end{align*}
Since $d^1_t(d^1_t-1)\leq d^2_t(d^2_t-1)$ as $d^1_t \leq d^2_t$ and both $d^1_t$ and $d^2_t$ are non-negative integers, \eqref{eq:subm_matroid_r} holds by the induction hypothesis.
The argument for~\eqref{eq:subm_matroid_b} is identical.

Next assume that $i = K-(t-1)$. Then~\eqref{eq:subm_matroid_r} holds, because the LHS and RHS of~\eqref{eq:subm_matroid_r} are equal to $2 m_{t-1} d^1_t$ and $2 m_{t-1} d^2_t$, respectively, and $d^1_t\leq d^2_t$.
For~\eqref{eq:subm_matroid_b}, first observe that, if $b^1_{K-(t-1)} \geq 2(t-1)$, then $d^1_t=0$, implying that the LHS of~\eqref{eq:subm_matroid_b} is zero, and hence~\eqref{eq:subm_matroid_b} is trivial since $f_t$ is monotone by Lemma~\ref{lem:matroidFunctionMonotonicity}.
So assume that $b^1_{K-(t-1)} < 2(t-1)$, implying that $d^2_t \geq d^1_t >0$.
Furthermore, by the monotonicity, we have $\delta^1_{t-1} \leq \delta^2_{t-1}$.
It follows from~\eqref{eq:ft} that the LHS and RHS of~\eqref{eq:subm_matroid_b} are equal to
\begin{align*}
2m_{t-1}s_t  + 2\delta^1_{t-1} (d^1_t-1) \mbox{\ and \ }
2m_{t-1}s_t  + 2\delta^2_{t-1} (d^2_t-1),
\end{align*}
respectively. Since $d^1_t\leq d^2_t$, this proves the lemma.
\end{proof}

We are left with proving indistinguishability.
 We begin by proving that when only elements of a particular class are present, the $f_t$-values are
entirely determined by their cardinality.

\begin{lemma}\label{lem:oneColorIndistiguishability}
For any $t=2,3,\dots, K$ and any non-negative integer $b$, it holds that
\[
f_t(1, 0,\ldots, 0; b, 0,\ldots,0) =f_t(0, 0,\ldots, 0; b+1, 0,\ldots, 0).
\]
\end{lemma}
\begin{proof} 
Let $d^1_t$ and $d^2_t$, $\delta^1_{t-1}$ and $\delta^2_{t-1}$ be the realizations of $d_t$ and $\delta_{t-1}$ based on $(1,0,\dots, 0;b,0,\dots, 0)$ and $(0,0,\dots, 0;b+1,0,\dots, 0)$, respectively.
Then
$\delta^1_{t-1}= \delta^2_{t-1} = m_{t-1}$ holds by Lemma~\ref{lem:matroidFunctionMonotonicity}.
By~\eqref{eq:ft}, the two function values stated in the lemma can be written as
\begin{align*}
f_t(1, 0,\ldots, 0; b\phantom{+1}, 0,\ldots,0) & = m_t - (d^1_t-1)d^1_t m_{t-1},\\
f_t(0, 0,\ldots, 0; b+1, 0,\ldots, 0)& = m_t - (d^2_t+1)d^2_t m_{t-1},
\end{align*}
respectively.
If $b \geq 2(t-1)$, then both $d^1_t$ and $d^2_t$ are 0, which gives the equivalence of the two above expressions.
So assume that $b+1 \leq 2(t-1)$, and hence $d^1_t = d^2_t+1$.
Then, since
\begin{align*}
(d^1_t-1)d^1_t = (d^2_t+1)(d^2_t+1-1) = (d^2_t+1)d^2_t,
\end{align*}
the above two expressions are equivalent. The proof follows.
\end{proof}

The next lemma generalizes the previous one: if we fix the numbers of red and blue elements of classes before class $i$, then
the function value is entirely determined by the cardinality of elements in class $i$.
This proves Lemma~\ref{lem:hard-function-partition}~(iv) by setting $t=K$.

\begin{lemma}\label{lem:ManyColorsIndistiguishability}
For any $i=2,3, \dots, K$ and any integer $t$ such that $K-i+1 \leq t \leq K$, it holds that
\begin{align*}
&f_t(r_{K-(t-1)}, \ldots, r_{i+1}, 1,0,\ldots,  0; b_{K-(t-1)},\ldots, b_{i+1}, b_i, 0,\ldots, 0)  \\
&= f_t(r_{K-(t-1)}, \ldots, r_{i+1}, 0,0,\ldots,  0; b_{K-(t-1)},\ldots, b_{i+1}, b_i+1, 0,\ldots, 0).
\end{align*}
\end{lemma}

\begin{proof} We prove by induction on $t$. The base case $t=K-i+1$ follows from Lemma~\ref{lem:oneColorIndistiguishability}.
For the induction step when $t > K-i+1$, let $\delta^\ell_{t-1}$ be the realization of $\delta_{t-1}$ based on $\bigcup_{j=K-(t-2)}^{i} \{ r^\ell_j, b^\ell_j\}$ for $\ell =1,2$.
Induction hypothesis then states that $\delta^1_{t-1}=\delta^2_{t-1}$ and
the proof follows by observing how the function is defined by recurrence~\eqref{eq:ft}.
\end{proof}

Lemma~\ref{lem:hard-function-partition} follows from Lemmas~\ref{lem:matroidFunctionMonotonicity},~\ref{lem:matroidFunctionOutput},~\ref{lem:matroidFunctionSubmodularity}, and~\ref{lem:ManyColorsIndistiguishability}.

\subsection{Proof of Observations in Section~\ref{sec:MatroidObs}}\label{sec:ProofObservations}

We first prove Lemma~\ref{lem:obs1}.

\begin{proof}[Proof of Lemma~\ref{lem:obs1}]
Let $t$ be an integer from 2 to $K$.
We first compute $f_t(1,0,\dots, 0; 0,\dots, 0)$.
By Lemma~\ref{lem:matroidFunctionMonotonicity}, $f_{t-1}(0,\dots, 0;0,\dots, 0) =0$ and hence $\delta_{t-1}=m_{t-1}$.
Hence, since $s_t=0$ and $d_t=2(t-1)$, it follows from~\eqref{eq:ft} that
\[
f_t(1,0,\dots, 0; 0,\dots, 0) = m_t - m_{t-1}(2s_t + d_t-1)d_t = m_t - m_{t-1}(2t-2)(2t-3) = 2\cdot (2t-2)!.
\]
On the other hand, since $s_t=1$ and $d_t=2(t-1)-1$ in the case of $f_t(0,\dots, 0; 1,0,\dots, 0)$,
\[
f_t(0,\dots, 0; 1,0,\dots, 0) = m_t - m_{t-1}(2s_t + d_t-1)d_t = m_t - m_{t-1}(2t-2)(2t-3) = 2\cdot (2t-2)!.
\]

Suppose that we are given $r_{K-(t-1)},r_{K-(t-2)},\ldots, r_{K}$ and $b_{K-(t-1)}, b_{K-(t-2)},\ldots, b_{K}$.
Let $f'$ be the value of $f_t(r_{K-(t-1)},r_{K-(t-2)},\ldots, r_{K}; b_{K-(t-1)}, b_{K-(t-2)},\ldots, b_{K})$.
Then
\begin{align*}
f_{t+1}&(0, r_{K-(t-1)},r_{K-(t-2)},\ldots, r_{K}; 0, b_{K-(t-1)}, b_{K-(t-2)},\ldots, b_{K})\\
&= m_{t+1} - (2m_{t} s_{t+1} + (m_{t}-f')(d_{t+1}-1)) d_{t+1} \\
& = (2t)(2t-1)\cdot f'.
\end{align*}
Therefore, when  $r_{K-(t-1)}=1$ and $r_{K-(t-2)}=\cdots=r_{K}=b_{K-(t-1)}= b_{K-(t-2)}=\cdots = b_{K}=0$~($2\leq t\leq K$), it follows that
\begin{align*}
f(0,\dots, 0, 1, 0, \dots, 0; 0, \dots, 0) &= (2K-2)(2K-3)\cdots (2t) (2t-1)f_t(1, 0, \dots, 0; 0, \dots, 0) \\
&  = (2K-2)(2K-3)\cdots (2t) (2t-1)\left(2\cdot (2t-2)!\right) = 2\cdot (2K-2)!.
\end{align*}
Similarly, we have $f(0, \dots, 0;0,\dots, 0, 1, 0, \dots, 0) =2\cdot (2K-2)!$.
Moreover, since $f_1(1;0)=1$ by the definition, it holds that
\[
f(0,\dots,0, 1; 0, \dots, 0)  = (2K-2)(2K-3)\cdots 2 \cdot f_1(1;0) = (2K-2)!.
\]
Thus Lemma~\ref{lem:obs1} holds.
\end{proof}

Below we prove Lemma~\ref{lem:obs2}.
\begin{proof}[Lemma~\ref{lem:obs2}]
It follows from Lemma~\ref{lem:matroidFunctionMonotonicity} that
\[
f(1,\dots, 1;0,\dots, 0) = f_K(1,\dots, 1;0,\dots, 0) = (2K-1)!.
\]
Moreover, Lemma~\ref{lem:obs2} imlplies that
\begin{align*}
&f(1,0,\dots, 0; 0, \dots, 0)+f(0,1,\dots, 0; 0, \dots, 0)+\dots +f(0,\dots,0,1; 0, \dots, 0) \\
&=(K-1)\cdot \left(2\cdot (2K-2)!\right)+(2K-2)! = (2K-1)!.
\end{align*}
Thus the first equality holds.

We next consider computing $f(0,\dots, 0;2(K-1),0,\dots, 0)$.
Since $f_{K-1}(0,\dots, 0;0,\dots, 0)=0$, we have $\delta_{K-1}=m_{K-1}$.
Hence, since $s_K=1$ and $d_K=0$ in this case, we see from~\eqref{eq:ft} that
\[
f_K(0,\dots, 0;2(K-1),0,\dots, 0) = m_K - m_{K-1}(2s_K + (d_K-1))d_K = m_K.
\]
Thus Lemma~\ref{lem:obs2} follows.
\end{proof}


\section{Algorithms}\label{sec:algo}

In this section, we present algorithms for a cardinality constraint and a matroid constraint, respectively.

We first establish some convention here. We call the input problem the \emph{original instance}, in which all elements in the stream
are considered. Furthermore,

\begin{itemize}
\itemsep=0pt
\item$K$ denotes the size constraint for the input cardinality constraint or the rank of the input matroid;
\item $\opt$ is the optimal solution of the original instance;
\item $f\colon 2^E \to \bbZ_+$ is the input submodular function.
\end{itemize}

All algorithms are given in the form of a general procedure with a set of parameters (in particular, a re-defined submodular function
and a modified cardinality/matroid constraint). Each invocation of the general procedure is called a \emph{branch}. A branch
considers all remaining elements that have not arrived so far (i.e., a suffix of the entire stream of elements).
For a branch, we use the following notation.


\begin{itemize}
\itemsep=0pt
\item$k$ denotes the size constraint for a new cardinality constraint or the rank of the new matroid (as a rule $k \leq K$);
\item $\optD$ is the intersection of $\opt$ and the remaining elements considered by this branch;
\item $g$ is a submodular function derived from $f$. More specifically, $g(T)= f(T \mid S) := f(T \cup S) - f(S)$, where $S$ is a subset of elements that have already arrived so far before this branch starts.
\end{itemize}

Whenever $\optD$ is non-empty, we denote by $o_1$ the first element that arrives in the stream among all elements in $\optD$.


In the description of the algorithms, we assume that we are given an approximate $v \in \bbR_+$ so that
$v \leq f(\optD)\leq (1+\varepsilon)v$.
and we will explain in Section~\ref{sec:Implementation} how to implement them without knowing the value $v$.

\subsection{Cardinality Constraint}\label{sec:cardinality_weak}

We start by defining the main procedure $\Call{Cardinality}{}$.
The procedure takes four parameters $k$, $s$, $v$, and $g$; we assume that  $\optD$ satisfies the conditions that $g(\optD) \geq v$ and $|\optD|\;\leq k$, and $s$ is an upper bound  that we impose on the size of the returned solution of $\Call{Cardinality}{k,s,v,f}$.

$\Call{Cardinality}{k,s,v,g}$ is described in Algorithm~\ref{alg:cardinality}. Its basic idea can be simply described as follows.
We target the approximation ratio of $\frac{s}{s+k}$, which, intuitively, states that the ratio should get better as the allowed
solution size $s$ is increased with respect to $k$, the upper bound on the size of $\optD$. In case
that $k$ or $s=1$, we simply choose the element that gives the largest $g$-value. So assume that $k, s \geq 2$.

Depending on the value of $g(o_1)$ (for which we have no prior knowledge), we create two branches:

\begin{itemize}

\item \textbf{Branch} 1: If $g(o_1)$ is sufficiently large (precisely at least $\frac{v}{k+s-1}$), we just take the first element $e$ so that $g(e) \geq \frac{v}{k+s-1}$.
Then $e$ precedes $o_1$ (or is just $o_1$) and the rest of $\optD$.
We define a new submodular function $g'= g(\cdot \mid e)$.
We then invoke $\Call{Cardinality}{k, s-1, v - g(e), g'}$.

\item \textbf{Branch} 2: If $g(o_1)$ is too small, we simply ignore it and invoke $\Call{Cardinality}{k-1, s, \frac{k+s-2}{k+s-1}v, g}$ directly.

\end{itemize}

The output is just the better outcome of the two branches.

\begin{algorithm}[t!]
  \caption{}\label{alg:cardinality}
  \begin{algorithmic}[1]
  \Procedure{Cardinality}{$k,s,v,g$}

  \If{$k \geq 2$ and $s \geq 2$}
    \State{\underline{\textbf{(Branch 1)}}}
    \State{\hspace*{0.15in} Let $e$ be the first element such that $g(e) \geq \frac{v}{k+s-1}$}.
    \State{\hspace*{0.15in} Define $g' = g(\cdot \mid e)$.}
    \State{\hspace*{0.15in} Apply $\Call{Cardinality}{k, s-1, v - g(e), g'}$ on all the elements after $e$.}
    \State{\hspace*{0.15in} Let $S'$ be the returned solution.}
    \State{\hspace*{0.15in} $S_1 := S'+e$.}

    \State{\underline{\textbf{(Branch 2)}}}
    \State{\hspace*{0.15in} Apply $\Call{Cardinality}{k-1, s, \frac{k+s-2}{k+s-1}v, g}$ on all the elements.}
    \State{\hspace*{0.15in} Let $S_2$ be the returned solution.}
    \If{$g(S_1)>g(S_2)$}
    \Return{$S_1$.}
    \Else{}
    \Return{$S_2$.}
    \EndIf{}
  \EndIf{}
  \If{$k =1$ or $s =1$}
      \Return{$\argmax_{e} g(e)$.}
  \EndIf{}
  \EndProcedure{}
  \end{algorithmic}
\end{algorithm}

\begin{lemma}\label{lem:cardinalityBasic}
  Suppose that $k\geq 1$ and $s\geq 1$.
  Suppose that there is a set $\optD$ in the input stream so that $g(\optD)\geq v$ and $|\optD|\;\leq k$.
  Then $\Call{Cardinality}{k,s,v,g}$ returns a solution $S$ such that $g(S) \geq \frac{s}{k+s-1}v$ and $|S|\; \leq s$.
\end{lemma}

\begin{proof}
  We begin by noting our assumption is that $|\optD|\;\leq k$ and $g(\optD)\geq v$.
  If $\optD = \emptyset$, then $v \leq f(\emptyset)$, implying that
  any solution satisfies the lemma.
  Therefore, in the following, we assume that $\optD \neq \emptyset$.

  We now prove by induction, first on $k$ and then on $s$. In the base case $k=1$, as the algorithm chooses an element $e$ maximizing $g(e)$,
  we have $g(e) \geq g(o_1) = g(\optD) \geq v$.

  For the induction step $k>1$, we apply induction on $s$. In the base case $s=1$, again the algorithm chooses an element $e$ maximizing $g(e)$,
  so $g(e) \geq g(o_1) \geq \frac{g(\optD)}{k} \geq \frac{v}{k}$, where the second inequality follows from submodularity. For the induction step $s>1$, the algorithm
  creates two branches. Now consider two possibilities.

  \begin{itemize}

  \item Suppose that $g(o_1) \geq \frac{v}{k+s-1}$.
  Then \textbf{Branch} 1 is bound to find an element $e$ with $g(e) \geq \frac{v}{k+s-1}$ and $e$ either precedes $o_1$ or is just $o_1$.
  As $g'(\optD) \geq g(\optD) - g(e) \geq v -g(e)$, we can then apply induction hypothesis on $\Call{Cardinality}{k, s-1, v-g(e), g'}$, which returns a solution $S'$ with $g'(S') \geq \frac{s-1}{k+s-2}\bigl(v-g(e)\bigr)$ and $|S'|\;\leq s-1$.
  Then
  \[
  g(S_1) =  g'(S') + g(e)  \geq \frac{s-1}{k+s-2}v + \frac{k-1}{k+s-2}g(e) \geq \frac{s}{k+s-1}v.
  \]
  Clearly, $|S_1|\; \leq s$.

  \item Suppose that $g(o_1) < \frac{v}{k+s-1}$. Then $g(\optD-o_1)\geq g(\optD) -g(o_1) \geq \frac{k+s-2}{k+s-1}v$, due to submodularity. By the induction hypothesis,
  $\Call{Cardinality}{k-1, s, \frac{k+s-2}{k+s-1}v, g}$ in \textbf{Branch} 2 returns a solution $S_2$ with $g(S_2) \geq \frac{s}{k+s-2}\frac{k+s-2}{k+s-1}v = \frac{s}{k+s-1}v$.
  Clearly, $|S_2|\; \leq s$.
  \end{itemize}

  Therefore, one of the two branches gives the desired solution. This finishes the induction step on $s$ and then also on $k$. The proof follows.
\end{proof}


\begin{theorem}\label{thm:cardinalityBasic}
  Suppose that $v \leq f(\opt)$.
  Then,the algorithm $\Call{Cardinality}{K,K,v,f}$ returns a solution $S$ with $f(S) \geq \frac{K}{2K-1}v $.
  The space complexity (for a fixed $v$) is $O\left(K 2^{2K}\right)$.
\end{theorem}
\begin{proof}
  The first part follows from Lemma~\ref{lem:cardinalityBasic}.
  For space requirement, let $\Gamma(k,s)$ denotes the space required for $\Call{Cardinality}{k,s,v,g}$ and let $c$ be some constant.
  Then it follows from the algorithm that (1) $\Gamma(k,s) \leq c$ when $k=1$ or $s=1$, and (2) $\Gamma(k,s) \leq \Gamma(k-1,s) + \Gamma(k, s-1)+c$ when $k>1$ and $s>1$.
  It is easy to verify this recurrence leads to  $\Gamma(k,s) \leq k 2^{k+s}c$. The proof follows.
\end{proof}

Therefore, if we are given $v$ such that $v \leq f(\opt)\leq (1+\varepsilon)v$, $\Call{Cardinality}{K,K,v,f}$ returns a solution $S$ with $f(S) \geq \left(\frac{K}{2K-1}-\varepsilon\right)f(\opt)$.
Our algorithm consists in invoking $\Call{Cardinality}{K,K,v,f}$ for $v$ in some interval, while updating the interval dynamically.
See Section~\ref{sec:Implementation} for details.

\subsection{Matroid Constraint}

Let $\mathcal{M} = (E, \I)$ be a given matroid, whose ground set $E$ is the entire stream of elements.
Assume that the rank of $\mathcal{M}$ is $K$.
We present the procedure $\Call{Matroid}{k,v,g, I}$ as Algorithm~\ref{alg:matroid}.
Again $g$ is the submodular function
defined over the remaining elements, among which $\optD$ satisfies the conditions that $g(\optD) \geq v$ and $|\optD|\; = k \leq K$.
Moreover, an independent set $I \in \I$ is also given as a part of the input; such a set $I$ should guarantee that $I \cup \optD \in \I$.

We now give the intuition of this procedure.
We guess out possible values of $g(o_1)$ in intervals of $\frac{v}{K^4}$~(between $0$ and $0.5v$).
For each possible interval $\left[\frac{bv}{K^4}, \frac{(b+1)v}{K^4}\right]$ of $g(o_1)$, we create a branch, and we could potentially take the first element $e$ so that $I+e \in \I$, $g(e) \geq \frac{bv}{K^4}$, and then define
$g' = g(\cdot \mid e)$ and invoke $\Call{Matroid}{k-1, (1-\frac{1}{K^4})v-2g(e), g', I + e}$ on the elements after $e$.
So far the idea is similar to the previous cardinality case.
The problem of this approach is that there is no guarantee that $(I + e) \cup (\optD-o_1) \in \I$, that is a required condition for the procedure $\Call{Matroid}{k-1, (1-\frac{1}{K^4})v-2g(e), g', I + e}$.

To remedy this issue, we introduce the following idea.
For each possible interval $\left[\frac{bv}{K^4}, \frac{(b+1)v}{K^4}\right]$ of $g(o_1)$, we create a set $T$, initialized as $\emptyset$.
Every time a new element $e$
arrives so that (1) $I \cup T + e \in \I$, and (2) $g(e) \geq \frac{bv}{K^4}$, we add $e$ into $T$ and create a new branch
$\Call{Matroid}{k-1, (1-\frac{1}{K^4})v-2g(e), g', I + e}$, where $g' = g(\cdot \mid e)$.
As we will show (see the proof of Lemma~\ref{lem:matroidBasic}),
at least one of the elements $e \in T$ satisfies the property $(I + e) \cup (\optD-o_1) \in \I$.
Apparently, $|T|\; \leq K$, so in total we have at most $K+1$ branches for each $b$.


\begin{algorithm}[t!]
  \caption{}\label{alg:matroid}
  \begin{algorithmic}[1]
  \Procedure{$\Call{Matroid}{k,v,g,I}$}{}
    \If{$k>1$}
      \State{$\beta := \max_{b \in \bbZ_+} \frac{b}{K^4} \leq \frac{v}{2}$}.
        \For{$0 \leq b  \leq \beta$}
        \State{$T := \emptyset$.}
          \While{$| I \cup T| < K $}
          \State{Let $e$ be the first element in the remaining stream satisfying (1) $I \cup T + e \in \I$, and (2) $g(e) \geq \frac{bv}{K^4}$.}
          \State{\underline{\textbf{Branch $(b, |T|+1)$}}}
            \State{\hspace*{0.15in} Define $g' = g(\cdot \mid e)$.}
            \State{\hspace*{0.15in} Apply $\Call{Matroid}{k-1, (1-\frac{1}{K^4})v-2g(e), g', I + e}$ on all the elements after $e$.}
            \State{\hspace*{0.15in} Let the returned solution be $S'$.}
            \State{\hspace*{0.15in} $S_{b, |T|+1} := e + S'$.}
           \State{$T := T+e$.}
         \EndWhile{}
       \EndFor{}
    \EndIf{}
      \State{\underline{\textbf{Branch $0$}}}
        \State{\hspace*{0.15in} $S_{0} := \argmax_{e, e\in \I} g(e)$.}
      \State{\Return{$\argmax\left\{ g(S)\mid S\in\{S_{b,j}\mid b=0,\dots, \beta, j=1,\dots,k\}\cup S_0\right\}$.}}
  \EndProcedure{}
  \end{algorithmic}
\end{algorithm}

The following fact is well-known.
\begin{proposition}\label{pro:simpleFact}
  Suppose that $C_1$ and $C_2$ are two circuits and $x \in C_1 \cap C_2$. Then, for any $y \in C_1 \backslash C_2$, there exists another circuit $C \subseteq C_1 \cup C_2 - x$ and $C \ni y$.
\end{proposition}

\begin{lemma}\label{lem:matroidFirst}
  Let $T= \left\{e_1, \ldots, e_{|T|}\right\}$.
  Suppose that $I \cup T\in \I$ and $I \cup \optD  \in \I$ holds.
  Furthermore, suppose that for each $e_i \in T$, $I \cup \optD + e_i$ contains a circuit $C_i$ so that $C_i \not \ni o_1$.
  Then $I \cup T + o_1 \in \I$.
\end{lemma}
\begin{proof}
  We prove by establishing a more general statement: there is no circuit $C^* \subseteq I \cup \optD \cup T $ so that $C^* \ni o_1$
  (this implies that the lemma as $I \cup T \in \I$).
  We proceed by contradiction. Assume that such a circuit $C^*$ does exist and we will establish the following claim.
  \begin{claim}
    Suppose that $C^* \ni o_1$ and $C^* \subseteq I \cup \optD \cup T$ and $C^* \cap T  \neq \emptyset $.
    Then, there exists another circuit $\overline{C} \subseteq I \cup \optD \cup T$, $\overline{C} \ni o_1$ and $|\overline{C} \cap T| < |C^* \cap T|$.
  \end{claim}
  \begin{proof}
    To prove the claim, assume that $e_i \in C^* \cap T$.
    Then by Proposition~\ref{pro:simpleFact}, we have a circuit $\overline{C} \subseteq C_i \cup C^* - e_i$ and $\overline{C} \ni o_1$
    (recall that $C_i$ is the circuit contained in $I \cup \optD + e_i$ and $C_i \not \ni o_1$).
     Clearly, $|\overline{C} \cap T | < |C^* \cap T|$.
  \end{proof}

  Now by this claim, we conclude that there is a circuit $\overline{C} \subseteq I \cup \optD$ and $\overline{C} \ni o_1$,
  a contradiction to the assumption that $I \cup \optD  \in \I$.
\end{proof}

\begin{lemma}\label{lem:matroidBasic}
  Suppose that there is a set $\optD$ in the input stream so that $g(\optD) \geq v$, $I \cup \optD \in \I$, and $|\optD|\; = k $.
  Then  $\Call{Matroid}{k, v, g, I}$ returns a solution $S$ so that $g(S) \geq \frac{1}{2}\Bigl(1-\frac{1}{2K-k}\Bigr)v$ and $ I \cup S \in \I$.
\end{lemma}
\begin{proof}
  We prove by induction on $k$.
  For the base case $k=1$,
  \textbf{Branch} $0$ is bound to get an element $e$ so that $g(e) \geq g(\optD) \geq v$.
  So assume that $k>1$.
  Then there exists $b$ with $0 \leq b \leq \beta$ such that $\frac{bv}{K^4} \leq g(o_1) \leq   \frac{(b+1)v}{K^4}$.
  Let $T= \left\{e_1, \ldots, e_{|T|}\right\}$ be the set collected for this $b$ at the moment immediately before $o_1$ arrives.
  There are two possibilities.
  \begin{itemize}
    \itemsep=0pt
    \item If some element $e_i \in T$
    satisfies the condition that either $ I \cup \optD + e_i$ is independent or contains a circuit $C_i$ and $C_i \ni o_1$, then
    $ I \cup \optD + e_i - o_1 \in \I$.
    Moreover, it holds that $g'(\optD-o_1) \geq g(\optD)-g(o_1) - g(e)\geq \left(1-\frac{1}{K^4}\right)v - 2g(e)$ as $g(o_1)\leq g(e)+\frac{v}{K^4}$.
    \item Otherwise, that is, if every element $e_i \in T$ satisfies the condition that $I \cup \optD + e_i$ contains a circuit $C_i$
    and  $C_i \not \ni o_1$, then by Lemma~\ref{lem:matroidFirst}, $I\cup T+o_1\in \I$, implying that $o_1$ will be added into $T$.
  \end{itemize}
  In both cases, we know that there exists an element $e \in T \cup \{o_1\}$ so that (1) $g(e) \geq \frac{bv}{K^4}$,
  (2) $g'(\optD-e) \geq \left(1-\frac{1}{K^4}\right)v - 2g(e)$, and (3) $ (I + e) \cup  (\optD  - o_1) \in \I$.
  Furthermore, all elements of $\optD - o_1$ are considered by $\Call{Matroid}{k-1, \left(1-\frac{1}{K^4}\right)v - 2g(e), g', I + e}$.
  Then by induction hypothesis, in \textbf{Branch}$(b, j)$ for some $j$, we obtain a solution
  \begin{align*}
  g(S_{b, j}) & \geq g(e) + \left(\left(1-\frac{1}{K^4}\right)v - 2g(e)\right)\frac{1}{2}\left( 1- \frac{1}{2K-k+1} \right) \\
  & \geq \frac{1}{2}\left( 1- \frac{1}{2K-k} \right)v.
  \end{align*}
  This finishes the induction step.

  The fact that the returned solution $S$ satisfies $S \cup I \in \I$ is easy to verify.
\end{proof}


\begin{theorem}\label{thm:matroidBasic}
  Suppose that $v \leq f(\opt)$.
  Then the algorithm $\Call{Matroid}{K, v,f, \emptyset}$ returns a solution $S$ guaranteeing that $f(S) \geq \left(\frac{1}{2}-\frac{1}{2K}\right)v$.
  The space complexity (for a fixed $v$) is $O\left(K^{5K+1}\right)$.
\end{theorem}
\begin{proof}
  The first part follows from Lemma~\ref{lem:matroidBasic}.
  For the space complexity, consider the branching tree with $\Call{Matroid}{K, v,f, \emptyset}$ as the root.
  By the algorithm, each node has at most $O(K^5)$ branches.
  Furthermore, the space required for such a node is $K+c$ for some constant $c$~(for each branch, we need to store an element).
  The depth of such a tree is at most $K$. So the total complexity is at most $O(K^{5K+1})$.
\end{proof}

\subsection{Implementation}\label{sec:Implementation}
We now explain how to implement the algorithms described in the preceding sections without knowing the optimal value $v$.
We adapt the dynamic-update technique in~\cite{Badanidiyuru:2014ib}.

We here explain the cardinality case.
The matroid case is analogous.
We will let $v$ be a number of the form ${\bigl(1+\varepsilon\bigr)}^i$ for some $i \in \bbZ$.
We observe that $\max_e f(e)\leq f(\opt)\leq K \max_e f(e)$.
Let $m$ be the maximum of $f(e)$ among all elements $e$ that
have arrived so far.
The algorithm $\Call{Cardinality}{K,K,v,f}$ is activated only when $\frac{m}{{(1+\varepsilon)}^2} \leq v \leq \frac{K m}{\varepsilon}$.
The critical observation is that, if $f(\opt)>\frac{K m}{\varepsilon}$, then
the set of optimal items $\opt' \setminus \opt$ that have arrived so far has the property that
\[
f(\opt') \leq \sum_{o_i \in \opt'} f(o_i) \leq K m \leq \varepsilon f(\opt).
\]
Therefore, the first time an element $e$ arrives so that $m:= f(e)$ and $m \leq f(\opt)\leq \frac{K m}{\varepsilon}$,
there exists a subset $\optD = \opt \backslash \opt'$ so that $f(\optD) \geq (1- \varepsilon) f(\opt)$.
This means that we can use $\optD$ instead of $\opt$ to perform the algorithm.
Then the interval $\left[\frac{m}{{(1+\varepsilon)}^2}, \frac{K m}{\varepsilon}\right]$ contains $v$ such that $v\leq f(\optD)\leq (1+\varepsilon)v$, since $m \leq f(\opt)\leq \frac{K m}{\varepsilon}$.
It follows from Theorem~\ref{thm:cardinalityBasic} that, using $O\left(K 2^{2K}\right)$ space, $\Call{Cardinality}{K,K,v,f}$ returns a solution $S$ with $f(S) \geq \frac{K}{2K-1}v$, implying that
\[
f(S)\geq \frac{K}{2K-1}(1-\varepsilon)f(\optD)\geq \frac{K}{2K-1}(1- O(\varepsilon)) f(\opt).
\]
The number of guesses for $v$ is equal to $O\left(\log_{1+\varepsilon}\left( \frac{K}{\varepsilon}\right)\right)=O\left(\frac{\log(K/\varepsilon)}{\varepsilon}\right)$.


Theorems~\ref{thm:alg-weak-oracle-cardinality} and~\ref{thm:alg-weak-oracle-matroid} follow from the preceding discussion and Theorems~\ref{thm:cardinalityBasic} and~\ref{thm:matroidBasic}.







\bibliographystyle{abbrv}
\bibliography{main}

\begin{thebibliography}{10}

\bibitem{Alon:2012em}
N.~Alon, I.~Gamzu, and M.~Tennenholtz.
\newblock Optimizing budget allocation among channels and influencers.
\newblock In {\em Proceedings of the 21st International Conference on World
  Wide Web (WWW)}, pages 381--388, 2012.

\bibitem{Badanidiyuru:2014ib}
A.~Badanidiyuru, B.~Mirzasoleiman, A.~Karbasi, and A.~Krause.
\newblock Streaming submodular maximization: massive data summarization on the
  fly.
\newblock In {\em Proceedings of the 20th ACM SIGKDD International Conference
  on Knowledge Discovery and Data Mining (KDD)}, pages 671--680, 2014.

\bibitem{Badanidiyuru:2013jc}
A.~Badanidiyuru and J.~Vondr{\'a}k.
\newblock Fast algorithms for maximizing submodular functions.
\newblock In {\em Proceedings of the 25th Annual ACM-SIAM Symposium on Discrete
  Algorithms (SODA)}, pages 1497--1514, 2013.

\bibitem{BalkanskiRS19}
E.~Balkanski, A.~Rubinstein, and Y.~Singer.
\newblock An exponential speedup in parallel running time for submodular
  maximization without loss in approximation.
\newblock In {\em Proceedings of the 30th Annual {ACM-SIAM} Symposium on
  Discrete Algorithms (SODA)}, pages 283--302, 2019.

\bibitem{Balkanski2018}
E.~Balkanski and Y.~Singer.
\newblock The adaptive complexity of maximizing a submodular function.
\newblock In {\em Proceedings of the 50th Annual ACM Symposium on Theory of
  Computing (STOC)}, pages 1138--1151, 2018.

\bibitem{Barbosa2015}
R.~Barbosa, A.~Ene, H.~L.~Nguy{\fontencoding{T5}\selectfont \~\ecircumflex{}}n,
  and J.~Ward.
\newblock The power of randomization: Distributed submodular maximization on
  massive datasets.
\newblock In {\em Proceedings of the 32nd International Conference on
  International Conference on Machine Learning (ICML)}, pages 1236--1244, 2015.

\bibitem{Barbosa2016}
R.~D.~P. Barbosa, A.~Ene, H.~L. Nguy{\fontencoding{T5}\selectfont
  \~\ecircumflex{}}n, and J.~Ward.
\newblock A new framework for distributed submodular maximization.
\newblock In {\em Proceedings of the IEEE 57th Annual Symposium on Foundations
  of Computer Science (FOCS)}, pages 645--654, 2016.

\bibitem{Bateni:2017}
M.~Bateni, H.~Esfandiari, and V.~Mirrokni.
\newblock Almost optimal streaming algorithms for coverage problems.
\newblock In {\em Proceedings of the 29th ACM Symposium on Parallelism in
  Algorithms and Architectures (SPAA)}, pages 13--23, 2017.

\bibitem{Calinescu:2011ju}
G.~Calinescu, C.~Chekuri, M.~P{\'a}l, and J.~Vondr{\'a}k.
\newblock Maximizing a monotone submodular function subject to a matroid
  constraint.
\newblock {\em SIAM Journal on Computing}, 40(6):1740--1766, 2011.

\bibitem{DBLP:journals/mp/ChakrabartiK15}
A.~Chakrabarti and S.~Kale.
\newblock Submodular maximization meets streaming: matchings, matroids, and
  more.
\newblock {\em Mathematical Programming}, 154(1-2):225--247, 2015.

\bibitem{ChanSODA2017}
T.-H.~H. Chan, Z.~Huang, S.~H.-C. Jiang, N.~Kang, and Z.~G. Tang.
\newblock Online submodular maximization with free disposal: Randomization
  beats for partition matroids online.
\newblock In {\em Proceedings of the 28th Annual {ACM-SIAM} Symposium on
  Discrete Algorithms (SODA)}, pages 1204--1223, 2017.

\bibitem{Chan2017}
T.-H.~H. Chan, S.~H.-C. Jiang, Z.~G. Tang, and X.~Wu.
\newblock Online submodular maximization problem with vector packing
  constraint.
\newblock In {\em Proceedings of the 25th Annual European Symposium on
  Algorithms (ESA)}, pages 24:1--24:14, 2017.

\bibitem{DBLP:conf/icalp/ChekuriGQ15}
C.~Chekuri, S.~Gupta, and K.~Quanrud.
\newblock Streaming algorithms for submodular function maximization.
\newblock In {\em Proceedings of the 42nd International Colloquium on Automata,
  Languages, and Programming (ICALP)}, pages 318--330, 2015.

\bibitem{ChekuriQ19}
C.~Chekuri and K.~Quanrud.
\newblock Submodular function maximization in parallel via the multilinear
  relaxation.
\newblock In {\em Proceedings of the 30th Annual {ACM-SIAM} Symposium on
  Discrete Algorithms (SODA)}, pages 303--322, 2019.

\bibitem{DBLP:journals/siamcomp/ChekuriVZ14}
C.~Chekuri, J.~Vondr{\'{a}}k, and R.~Zenklusen.
\newblock Submodular function maximization via the multilinear relaxation and
  contention resolution schemes.
\newblock {\em {SIAM} Journal on Computing}, 43(6):1831--1879, 2014.

\bibitem{ene_et_al:LIPIcs:2019:10629}
A.~Ene and H.~L. Nguy{\fontencoding{T5}\selectfont \~\ecircumflex{}}n.
\newblock A nearly-linear time algorithm for submodular maximization with a
  knapsack constraint.
\newblock In {\em Proceedings of the 46th International Colloquium on Automata,
  Languages, and Programming (ICALP)}, volume 132, pages 53:1--53:12, 2019.

\bibitem{EneN19}
A.~Ene and H.~L. Nguy{\fontencoding{T5}\selectfont \~\ecircumflex{}}n.
\newblock Submodular maximization with nearly-optimal approximation and
  adaptivity in nearly-linear time.
\newblock In {\em Proceedings of the 30th Annual {ACM-SIAM} Symposium on
  Discrete Algorithms (SODA)}, pages 274--282, 2019.

\bibitem{ene_et_al:LIPIcs:2019:10630}
A.~Ene and H.~L. Nguy{\fontencoding{T5}\selectfont \~\ecircumflex{}}n.
\newblock Towards nearly-linear time algorithms for submodular maximization
  with a matroid constraint.
\newblock In {\em Proceedings of the 46th International Colloquium on Automata,
  Languages, and Programming (ICALP)}, volume 132, pages 54:1--54:14, 2019.

\bibitem{Feige:1998gx}
U.~Feige.
\newblock A threshold of $\ln n$ for approximating set cover.
\newblock {\em Journal of the ACM}, 45(4):634--652, 1998.

\bibitem{Filmus:2014}
Y.~Filmus and J.~Ward.
\newblock A tight combinatorial algorithm for submodular maximization subject
  to a matroid constraint.
\newblock {\em SIAM Journal on Computing}, 43(2):514--542, 2014.

\bibitem{FNS_cardinality}
M.~L. Fisher, G.~L. Nemhauser, and L.~A. Wolsey.
\newblock An analysis of approximations for maximizing submodular set functions
  {I}.
\newblock {\em Mathematical Programming}, pages 265--294, 1978.

\bibitem{FisherNemhauserWolsey}
M.~L. Fisher, G.~L. Nemhauser, and L.~A. Wolsey.
\newblock An analysis of approximations for maximizing submodular set functions
  {II}.
\newblock {\em Mathematical Programming Study}, 8:73--87, 1978.

\bibitem{HuangKakimuraMultiPass2018}
C.-C. Huang and N.~Kakimura.
\newblock Multi-pass streaming algorithms for monotone submodular function
  maximization.
\newblock {\em CoRR}, abs/1802.06212, 2018.

\bibitem{HuangKakimuraWADS2019}
C.-C. Huang and N.~Kakimura.
\newblock Improved streaming algorithms for maximizing monotone submodular
  functions under a knapsack constraint.
\newblock In {\em Proceedings of the Algorithms and Data Structures Symposium
  (WADS)}, pages 438--451, 2019.

\bibitem{Huang2019}
C.-C. Huang, N.~Kakimura, and Y.~Yoshida.
\newblock Streaming algorithms for maximizing monotone submodular functions
  under a knapsack constraint.
\newblock {\em Algorithmica}, 2019.

\bibitem{Kazemi2019}
E.~Kazemi, M.~Mitrovic, M.~Zadimoghaddam, S.~Lattanzi, and A.~Karbasi.
\newblock Submodular streaming in all its glory: Tight approximation, minimum
  memory and low adaptive complexity.
\newblock In {\em Proceedings of the 36th International Conference on Machine
  Learning (ICML)}, pages 3311--3320, 2019.

\bibitem{Kempe:2003iu}
D.~Kempe, J.~Kleinberg, and {\'E}.~Tardos.
\newblock Maximizing the spread of influence through a social network.
\newblock In {\em Proceedings of the 9th ACM SIGKDD International Conference on
  Knowledge Discovery and Data Mining (KDD)}, pages 137--146, 2003.

\bibitem{Krause:2008vo}
A.~Krause, A.~P. Singh, and C.~Guestrin.
\newblock Near-optimal sensor placements in gaussian processes: Theory,
  efficient algorithms and empirical studies.
\newblock {\em Journal of Machine Learning Research}, 9:235--284, 2008.

\bibitem{Kulik:2013ix}
A.~Kulik, H.~Shachnai, and T.~Tamir.
\newblock Maximizing submodular set functions subject to multiple linear
  constraints.
\newblock In {\em Proceedings of the 20th Annual ACM-SIAM Symposium on Discrete
  Algorithms (SODA)}, pages 545--554, 2013.

\bibitem{Kumar:2015}
R.~Kumar, B.~Moseley, S.~Vassilvitskii, and A.~Vattani.
\newblock Fast greedy algorithms in {MapReduce} and streaming.
\newblock {\em ACM Transactions on Parallel Computing}, 2(3):14:1--14:22, 2015.

\bibitem{Lee:2006cm}
J.~Lee.
\newblock {\em Maximum Entropy Sampling}, volume~3 of {\em Encyclopedia of
  Environmetrics}, pages 1229--1234.
\newblock John Wiley {\&} Sons, Ltd., 2006.

\bibitem{Lee:2010}
J.~Lee, M.~Sviridenko, and J.~Vondr{\'a}k.
\newblock Submodular maximization over multiple matroids via generalized
  exchange properties.
\newblock {\em Mathematics of Operations Research}, 35(4):795--806, 2010.

\bibitem{Lin:2010wpa}
H.~Lin and J.~Bilmes.
\newblock Multi-document summarization via budgeted maximization of submodular
  functions.
\newblock In {\em Proceedings of the 2010 Annual Conference of the North
  American Chapter of the Association for Computational Linguistics: Human
  Language Technologies (NAACL-HLT)}, pages 912--920, 2010.

\bibitem{Lin:2011wt}
H.~Lin and J.~Bilmes.
\newblock A class of submodular functions for document summarization.
\newblock In {\em Proceedings of the 49th Annual Meeting of the Association for
  Computational Linguistics: Human Language Technologies (ACL-HLT)}, pages
  510--520, 2011.

\bibitem{McGregor2019}
A.~McGregor and H.~T. Vu.
\newblock Better streaming algorithms for the maximum coverage problem.
\newblock {\em Theory of Computing Systems}, 63(7):1595, 2019.

\bibitem{Nemhauser:1978dm}
G.~L. Nemhauser and L.~A. Wolsey.
\newblock Best algorithms for approximating the maximum of a submodular set
  function.
\newblock {\em Mathematics of Operations Research}, 3(3):177--188, 1978.

\bibitem{NorouziFard:2018vb}
A.~Norouzi-Fard, J.~Tarnawski, S.~Mitrovic, A.~Zandieh, A.~Mousavifar, and
  O.~Svensson.
\newblock Beyond $1/2$-approximation for submodular maximization on massive
  data streams.
\newblock In {\em Proceedings of the 35th International Conference on Machine
  Learning (ICML)}, pages 3826--3835, 2018.

\bibitem{Soma:2014tp}
T.~Soma, N.~Kakimura, K.~Inaba, and K.~Kawarabayashi.
\newblock Optimal budget allocation: Theoretical guarantee and efficient
  algorithm.
\newblock In {\em Proceedings of the 31st International Conference on Machine
  Learning (ICML)}, pages 351--359, 2014.

\bibitem{Vondrak:2013ia}
J.~Vondr{\'a}k.
\newblock Symmetry and approximability of submodular maximization problems.
\newblock {\em SIAM Journal on Computing}, 42(1):265--304, 2013.

\bibitem{Wolsey:1982}
L.~Wolsey.
\newblock Maximising real-valued submodular functions: primal and dual
  heuristics for location problems.
\newblock {\em Mathematics of Operations Research}, 1982.

\bibitem{yoshida_2018}
Y.~Yoshida.
\newblock Maximizing a monotone submodular function with a bounded curvature
  under a knapsack constraint.
\newblock {\em SIAM Journal on Discrete Mathematics}, 33(3):1452--1471, 2018.

\bibitem{Yu:2016}
Q.~Yu, E.~L. Xu, and S.~Cui.
\newblock Streaming algorithms for news and scientific literature
  recommendation: Submodular maximization with a $d$-knapsack constraint.
\newblock {\em IEEE Global Conference on Signal and Information Processing
  (GlobalSIP)}, pages 1295--1299, 2016.

\end{thebibliography}

\end{document}